\documentclass[11pt]{article}

\usepackage[margin=1in]{geometry}
\usepackage{amsmath,amssymb,amsthm,mathtools}
\usepackage{bm}
\usepackage{enumitem}
\usepackage{float}
\usepackage{hyperref}
\usepackage{natbib}
\usepackage{xcolor}

\newtheorem{theorem}{Theorem}
\newtheorem{proposition}{Proposition}
\newtheorem{lemma}{Lemma}
\newtheorem{corollary}{Corollary}
\theoremstyle{definition}
\newtheorem{assumption}{Assumption}
\newtheorem{remark}{Remark}

\newcommand{\E}{\mathbb{E}}
\newcommand{\PP}{\mathbb{P}}
\newcommand{\R}{\mathbb{R}}
\newcommand{\Var}{\operatorname{Var}}

\newcommand{\one}{\mathbf{1}}
\newcommand{\cS}{\mathcal{S}}
\newcommand{\cH}{\mathcal{H}}
\newcommand{\cA}{\mathcal{A}}
\DeclareMathOperator*{\argmin}{arg\,min}

\newcommand{\cC}{\mathcal{C}}

\newcommand{\xfit}{\hat g^{(-\ell)}}

\title{\textbf{Debiased and Simultaneous Inference for\\
Heterogeneous Factorial Effect Modifiers\\
via Residualized Walsh--Hadamard Scores}}

\author{Tomoshige Nakamura \and Ryo Emoto%
\thanks{Juntendo University, Kyoto University}
}

\date{\today}

\begin{document}
\maketitle

\begin{abstract}
In a factorial randomized experiment with $K$ binary treatment components, scientific interest
often lies not in a single average effect but in which baseline covariates modify which
component-level main effects and interactions. We study the family of conditional
Walsh--Hadamard contrasts $\{\tau_S(x):\emptyset\neq S\subseteq[K]\}$ and the associated grid of
\emph{effect-modifier cells} $(j,S)$, indexed by covariates $j\in[p]$ and contrasts $S$ on the
Boolean subset lattice, through the coefficients $\theta^\ast_{jS}$ of the population linear
projection of $\tau_S$ onto standardized covariates. Building on a baseline-robust direct-score
identification of $\tau_S(x)$, we develop frequentist inference over all $p(2^K-1)$ cells. First,
we prove a debiased central limit theorem whose distinctive feature is \emph{global
insensitivity}: the residualization baseline enters the cross-fitted score through an exactly
conditionally mean-zero perturbation, so the influence-function remainder requires only foldwise
weighted $L_2$ consistency of the baseline estimator, with no prescribed polynomial rate and no
nuisance product-rate condition; a separate, explicit penalty-scale condition governs the
baseline's effect on the score Lasso. Second, under explicitly stated uniform nuisance, baseline,
studentization, and influence-array conditions---all derived from bounded sparse-regression
primitives---we establish a high-dimensional Gaussian approximation over the full grid,
calibrated by a Walsh-spectral multiplier bootstrap. This yields simultaneous confidence bands
and, via a Romano--Wolf step-down, \emph{selection of modified cells with strong familywise error
control}, without any effect-heredity assumption. The factorial design is fixed in the
asymptotic analysis, while the covariate dimension may grow; thus all $2^K-1$ contrasts are
covered jointly as the number of covariate-by-contrast cells increases. 
\end{abstract}


\section{Introduction}\label{sec:intro}

Factorial randomized experiments vary several treatment components at once and can identify
component-level main effects and interactions efficiently.  In many applications, however, the
scientific question is not only which factorial effects are nonzero on average, but also which
pretreatment covariates modify which effects.  With $K$ binary components and $p$ covariate
features, this question generates a grid of $p(2^K-1)$ covariate-by-contrast cells.  We study the
coefficient $\theta_{jS}^\ast$ in the population linear projection of the heterogeneous factorial
effect $\tau_S(X)$ on standardized covariates, for every covariate $j\in[p]$ and every nonempty
contrast $S\subseteq[K]$.  The inferential objectives are a pointwise interval for a
prespecified cell, a simultaneous confidence band over the full grid, and a selection rule that
controls familywise error when the grid is searched for effect modification.

The simultaneous problem is essential rather than cosmetic.  Investigators commonly examine
several baseline characteristics across factorial main effects and interactions before deciding
which cells to emphasize.  Pointwise intervals do not protect this data-dependent search: even
under a global modifier null, the probability of at least one chance finding increases with the
number of cells.  A simultaneous band instead provides one joint uncertainty statement for the
entire covariate-by-contrast array, while a procedure with strong familywise error rate (FWER)
control permits confirmatory selection under any configuration of true and false modifier nulls.  The challenge is to obtain these
guarantees when the array is high-dimensional, the effect-modifier regressions are sparse, and a
flexible outcome baseline is estimated from the same experiment.

Our starting point is a residualized inverse-probability weighted Walsh--Hadamard score.  A
direct plug-in analysis would estimate $2^K$ arm-specific outcome regressions and then contrast
them, leaving each regression with an effective sample size of roughly $n/2^K$ under balanced
assignment.  The proposed score instead combines every observation for every contrast.  Its key
structural property is exact: under known randomized assignment, the zero-sum Walsh identity
makes the conditional mean of the score equal to $\tau_S(X)$ after subtracting \emph{any}
treatment-invariant baseline.  The baseline therefore does not alter identification, although
an accurate baseline can reduce variance substantially.  We estimate the baseline by
cross-fitting, regress the resulting contrast-specific scores on a common covariate design, and
de-bias the Lasso coefficients using nodewise residualization computed once for that design.

The same randomization identity also drives the asymptotic theory.  Conditional on the training
sample for a fold, the perturbation created by replacing a limiting baseline with its
cross-fitted estimate is exactly mean zero for all contrasts and covariates.  This \emph{global
insensitivity} differs from an ordinary local orthogonality expansion.  The influence-function
remainder requires foldwise weighted consistency of the baseline estimator, but no prescribed
polynomial convergence rate and no product-rate condition; a separate, explicit condition keeps
the baseline-induced empirical score below the Lasso penalty scale.  After pointwise
de-biasing, we strengthen the representation uniformly over all $p(2^K-1)$ cells, apply a
high-dimensional Gaussian approximation, and estimate the law of the maximum with a multiplier
bootstrap.  The resulting band is jointly valid, and a Romano--Wolf step-down based on the same
multiplier draws controls the strong familywise error rate without imposing effect heredity on
the Boolean lattice.

The paper makes four contributions.  First, it gives a direct-score identification and
estimation strategy for every heterogeneous factorial contrast that uses the full randomized
sample and separates the validity of residualization from its efficiency.  Second, it
establishes pointwise asymptotic linearity, feasible sandwich inference, and the foldwise global
centering argument underlying the weak baseline-rate requirement.  Third, it develops uniform
linearization, Gaussian and multiplier-bootstrap approximation, simultaneous confidence bands,
and heredity-free strong-FWER selection over the full modifier grid.  The factorial dimension is
fixed in the asymptotic analysis, whereas $p=p_n$ and hence the number of cells may grow; the
supplementary material derives the high-level uniform conditions from bounded sparse-regression
primitives.  Fourth, it provides an implementable algorithm in which the common nodewise fit and
multiplier draws are reused across contrasts and step-down iterations.

The simulations clarify both the scope and the finite-sample content of these results.  Across
primary designs with $300$ to $3150$ simultaneous cells, coverage and Romano--Wolf FWER remain
close to their nominal levels as the factorial dimension, covariate dimension, and sample size
vary.  A cross-fitted quadratic Lasso baseline approaches oracle precision as $n$ increases and,
at $n=1000$, produces bands within about $2\%$ of the oracle width when prognostic variation is
present.  Misspecified treatment-invariant baselines remain approximately calibrated but can
lose substantial power, illustrating the distinction between validity and efficiency.  By
contrast, the diagnostic that fits the quadratic baseline without cross-fitting has coverage
$0.902$ at $n=200$ despite its apparently shorter intervals, showing that sample splitting is
substantively important.  The multiplier calibration is only modestly shorter and more powerful
than Bonferroni--Holm in this particular Gaussian design, so our empirical claim concerns
calibration and baseline adaptivity rather than a universal numerical dominance over
dependence-agnostic procedures.

The scope of the paper is deliberately specific.  The target $\theta_{jS}^\ast$ is a projection
coefficient relative to the chosen feature dictionary and covariate distribution; it is not in
general a derivative or a fully nonparametric conditional effect.  The theory assumes known
randomized assignment and fixed $K$, with growth driven by the covariate dimension.  Finally, we
claim strong familywise error control, not false discovery rate control.  Step-up FDR procedures
probe farther into the tails than the additive Gaussian-approximation bound used here can
justify; the required moderate-deviation analysis is left to a companion paper.

\subsection{Related work}\label{sec:related-work}

Causal inference for factorial treatments has developed along several complementary lines.
\citet{DasguptaPillaiRubin2015} formulated average factorial effects in the potential-outcomes
framework, and \citet{ZhaoDing2022} established design-based properties of factor-based
regression and robust Wald inference for such effects.  \citet{EgamiImai2019} developed the
average marginal interaction effect and regularized selection through a penalized ANOVA with
zero-sum constraints.  These methods concern population-level treatment-factor effects and
interactions rather than a full array of baseline-covariate modifiers.  More directly focused on
heterogeneity, \citet{Goplerud2025} used a Bayesian mixture of regularized logistic regressions
to discover groups with different effect patterns in high-dimensional factorial treatments;
the target and inferential guarantee differ from simultaneous frequentist coverage and
familywise-valid selection of covariate-by-contrast projection coefficients.

The factorial literature also contains important multiplicity corrections, but for different
objects.  \citet{KimelEtAl2008} developed FDR methods for screening active effects in
unreplicated fractional factorial experiments, and \citet{LinEtAl2016} derived simultaneous
inference for a prespecified set of treatment comparisons in $2\times2$ factorial survival
trials.  Conversely, \citet{SemenovaEtAl2023} developed simultaneous inference for a
high-dimensional conditional average treatment effect (CATE) coefficient vector outside the
factorial setting.  Our contribution lies
at the intersection: simultaneous bands and strong-FWER selection for the growing
$p(2^K-1)$ grid that crosses baseline covariates with all nonempty factorial contrasts.  This is
a narrower claim than saying that simultaneous inference for factorial designs is unavailable
in general.

Methodologically, transformed-outcome and orthogonal CATE estimators
\citep{Tian2014,NieWager2021,Kennedy2023,FosterSyrgkanis2023} provide the binary-treatment
analogue of the residualized score.  Known factorial randomization lets us specialize this
construction to the complete Walsh--Hadamard contrast family and yields exact baseline
insensitivity, rather than a nuisance product-rate remainder.  Desparsified Lasso inference
\citep{ZhangZhang2014,vandeGeer2014,JavanmardMontanari2014} supplies the coordinatewise
de-biasing step.  High-dimensional Gaussian and multiplier-bootstrap results
\citep{CCK2013,CCK2017,CCKComparison2013} supply approximation tools for the maximum, and the
step-down construction follows \citet{RomanoWolf2005}.  The present work combines these tools
with the shared factorial score and verifies the additional uniform remainder and
studentization conditions needed for the covariate-by-contrast lattice.

\paragraph{Organization.}
Section~2 defines the factorial contrasts, establishes direct-score identification and the
projection target, and constructs the cross-fitted scores.  Section~3 develops the foldwise
centering lemma, debiased pointwise inference, and consistent sandwich variance estimation.
Section~4 proves uniform linearization and Gaussian approximation, establishes multiplier-
bootstrap simultaneous bands and Romano--Wolf strong-FWER control, and gives a practical
implementation algorithm; the bounded-envelope primitive verification is provided in the
supplementary material.  Section~5 studies finite-sample calibration, dimensional and
sample-size scaling, baseline quality, and the role of cross-fitting.  Section~6 discusses the
scope of the results and extensions to clustered or observational designs, structure sharing,
nonlinear modification, and heavier-tailed settings.

\section{Setup: factorial RCT, direct score, projection target}\label{sec:setup}

\subsection{Observed data, factorial contrasts, and assumptions}
For each $n$, we observe an i.i.d.\ sample $\{(X_i,A_i,Y_i)\}_{i=1}^n$, where
$X_i\in\mathcal X\subseteq\R^p$ is a vector of pretreatment covariates,
$A_i=(A_{i1},\dots,A_{iK})\in\cA:=\{-1,+1\}^K$ is the assigned treatment combination, and
$Y_i\in\R$ is the observed outcome. Let $Y_i(a)$ denote the potential outcome under
$a\in\cA$, let $P_X$ denote the marginal law of $X$, and define the conditional potential mean
$\mu_a(x):=\E\{Y(a)\mid X=x\}$.

\begin{assumption}[Factorial RCT]\label{as:rct}
\emph{(Consistency)} $Y_i=Y_i(A_i)$; \emph{(known randomized assignment)}
$\PP\{A_i=a\mid X_i,\{Y_i(b)\}_{b\in\cA}\}=p_a$ with known $p_a$ and
$\sum_{a\in\cA}p_a=1$; \emph{(positivity)} $p_a>0$ for every $a\in\cA$;
\emph{(moments)} $\E\{|Y(a)|\mid X=x\}<\infty$ for every $a\in\cA$ and $P_X$-almost every
$x$.
\end{assumption}

For $S\subseteq[K]:=\{1,\dots,K\}$, define the Walsh--Hadamard basis function
$\phi_S(a):=\prod_{k\in S}a_k$, with $\phi_\emptyset\equiv1$. Let
$\cS:=\{S\subseteq[K]:S\neq\emptyset\}$ and $M:=|\cS|=2^K-1$. For $S\in\cS$, the
\emph{heterogeneous factorial effect} at $x$ is
\begin{equation}\label{eq:tau}
\tau_S(x):=2^{-(K-1)}\sum_{a\in\cA}\phi_S(a)\,\mu_a(x).
\end{equation}
It is a marginalized conditional main effect when $|S|=1$, a conditional
difference-in-differences interaction when $|S|=2$, and a higher-order component interaction
when $|S|>2$. The Walsh basis is orthonormal under the uniform measure on $\cA$ and satisfies
the zero-sum property
\begin{equation}\label{eq:zerosum}
\sum_{a\in\cA}\phi_S(a)=2^K\,\one\{S=\emptyset\}.
\end{equation}

\subsection{Residualized Walsh--Hadamard direct score}
A \emph{baseline} is any measurable function $g:\mathcal X\to\R$ of the pretreatment
covariates only. For any baseline $g$, define
\begin{equation}\label{eq:score}
\psi^g_S(Y,A,X):=2^{-(K-1)}\{Y-g(X)\}\,\frac{\phi_S(A)}{p_A}.
\end{equation}
For later use, write $q_S(a):=2^{-(K-1)}\phi_S(a)/p_a$, so that
$\psi^g_S=\{Y-g(X)\}q_S(A)$.

\begin{theorem}[Direct identification]\label{thm:ident}
Under Assumption~\ref{as:rct}, for any baseline $g$ and any $S\in\cS$,
$\E\{\psi^g_S\mid X=x\}=\tau_S(x)$.
\end{theorem}
\begin{proof}
Conditioning on $X=x$ and using $\PP(A=a\mid X=x)=p_a$,
$\E\{Y\phi_S(A)/p_A\mid x\}=\sum_a\phi_S(a)\mu_a(x)=2^{K-1}\tau_S(x)$, while
$\E\{\phi_S(A)/p_A\mid x\}=\sum_a\phi_S(a)=0$ for $S\neq\emptyset$ by \eqref{eq:zerosum};
the baseline term vanishes for \emph{any} $g$.
\end{proof}

The baseline does not affect identification; it only changes the conditional variance of
\eqref{eq:score}, whose minimizer is explicit.

\begin{proposition}[Variance-optimal baseline]\label{prop:gopt}
Suppose Assumption~\ref{as:rct} holds and, for every $a\in\cA$,
$\E\{Y(a)^2\mid X=x\}<\infty$ for $P_X$-almost every $x$. Fix $S\in\cS$ and write
$v_a(x):=\Var\{Y(a)\mid X=x\}$. Then, for every baseline $g$ and
$P_X$-almost every $x$,
\begin{equation}\label{eq:psi-second-moment}
\E\{(\psi^g_S)^2\mid X=x\}
=2^{-2(K-1)}\sum_{a\in\cA}p_a^{-1}\bigl[v_a(x)+\{\mu_a(x)-g(x)\}^2\bigr],
\end{equation}
\end{proposition}
\begin{proof}
By Assumption~\ref{as:rct}, conditional on $X=x$ the assignment $A$ is independent of
$\{Y(a)\}_{a\in\cA}$ with $\PP(A=a\mid X=x)=p_a$. Hence
$\E\{(\psi^g_S)^2\mid X=x\}=\sum_{a\in\cA}p_a\,q_S(a)^2\,\E[\{Y(a)-g(x)\}^2\mid X=x]$; since
$\phi_S(a)^2=1$ we have $p_a\,q_S(a)^2=2^{-2(K-1)}p_a^{-1}$, and the conditional
bias--variance decomposition
$\E[\{Y(a)-g(x)\}^2\mid X=x]=v_a(x)+\{\mu_a(x)-g(x)\}^2$ gives
\eqref{eq:psi-second-moment}.
\end{proof}

The right-hand side of \eqref{eq:psi-second-moment} does not depend on $S$. Moreover,
$\E(\psi^g_S\mid X=x)=\tau_S(x)$ is $g$-free by Theorem~\ref{thm:ident}, so the proposition
immediately implies that, for $P_X$-almost every $x$, $\Var(\psi^g_S\mid X=x)$ is minimized
over $g(x)\in\R$ at
\[
g_{\mathrm{opt}}(x)=\frac{\sum_{a\in\cA}p_a^{-1}\mu_a(x)}{\sum_{a\in\cA}p_a^{-1}}.
\]
In the balanced design $p_a\equiv2^{-K}$, this reduces to the uniform average
$2^{-K}\sum_{a\in\cA}\mu_a(x)$. Efficiency, not validity, is what the baseline buys; the
effect of $g_{\mathrm{opt}}$ on the asymptotic variance of the debiased estimator is recorded in
Remark~\ref{rem:bal} below.

\subsection{Projection target}
Let $Z(x)\in\R^p$ collect standardized non-intercept covariate features, so that
$\E\{Z(X)\}=0$ and $\E\{Z_j(X)^2\}=1$, and set $W(x):=(1,Z(x)^\top)^\top$. For each
$S\in\cS$, define the population linear projection
\begin{equation}\label{eq:projection-target}
\beta^\ast_S=(\alpha_S^\ast,\theta_S^{\ast\top})^\top
\in\argmin_{\beta\in\R^{p+1}}
\E\bigl[\{\tau_S(X)-W(X)^\top\beta\}^2\bigr].
\end{equation}
By Theorem~\ref{thm:ident}, the same coefficient minimizes
$\E[\{\psi^g_S-W(X)^\top\beta\}^2]$ for any baseline $g$ and is therefore $g$-free. We write
the coordinates of its non-intercept part as
$\theta^\ast_S=(\theta^\ast_{jS})_{j\in[p]}$. For a fixed contrast
$S$, the coordinate $\theta^\ast_{jS}$ is the coefficient on the standardized feature $Z_j$ in
the best linear approximation to $\tau_S(X)$. It therefore summarizes how the conditional
factorial effect for contrast $S$ varies linearly with feature $j$, after linearly accounting
for the other features in $Z$. The collection of these coefficients forms a $p\times M$ grid.
We call its index pair $(j,S)\in[p]\times\cS$ an \emph{effect-modifier cell} and, for each
cell, consider the null hypothesis
\[
H_{jS}:\theta^\ast_{jS}=0.
\]
Under this null, feature $j$ makes no linear contribution to the projected heterogeneity of
contrast $S$ after linear adjustment for the other included features. This does not rule out
nonlinear modification by feature $j$ or heterogeneity associated with other features. The
projection residual is $\varepsilon^g_S:=\psi^g_S-W(X)^\top\beta^\ast_S$, with
$\E\{W(X)\varepsilon^g_S\}=0$.

\subsection{Cross-fitted scores}
Split $[n]$ into folds $I_1,\dots,I_L$, and write $I_{-\ell}:=[n]\setminus I_\ell$. Using the
observations in $I_{-\ell}$, fit a baseline $\xfit$ and set, for $i\in I_\ell$,
\begin{equation}\label{eq:cfscore}
\hat\psi_{iS}=2^{-(K-1)}\{Y_i-\xfit(X_i)\}\,\frac{\phi_S(A_i)}{p_{A_i}} .
\end{equation}

\subsection{Notation for asymptotic analysis}
For a vector $u$, $\|u\|_1$, $\|u\|_2$, and $\|u\|_\infty$ are the $\ell_1$, Euclidean, and
maximum norms; $\|u\|_0$ is the number of nonzero entries. For a measurable $f$ and $q\ge1$,
$\|f\|_{L_q(P_X)}:=\{\E|f(X)|^q\}^{1/q}$ and
$\|f\|_\infty:=\sup_{x\in\mathcal X}|f(x)|$. For an array $a=(a_i)_{i=1}^n$,
$\|a\|_n:=(n^{-1}\sum_{i=1}^na_i^2)^{1/2}$ is the empirical $L_2$ norm. Write
$n_\ell:=|I_\ell|$, $\mathbb P_nf:=n^{-1}\sum_{i=1}^nf_i$, and
$\mathbb P_{n,\ell}f:=n_\ell^{-1}\sum_{i\in I_\ell}f_i$. All $o_p(\cdot)$ and $O_p(\cdot)$
statements are under the joint law of the sample, the fold partition (when random), and all
fitted nuisance functions, along a triangular array in which the covariate dimension
$p=p_n$ may grow. The number of treatment components $K$ and the allocation probabilities
$\{p_a:a\in\cA\}$ are fixed by the experimental design and do not vary with $n$. For positive
sequences, $a_n\lesssim b_n$ means $a_n\le Cb_n$ and
$a_n\asymp b_n$ means
$cb_n\le a_n\le Cb_n$ for constants $0<c\le C<\infty$ independent of $n$ and $p$ (and of
the target cell), but possibly depending on the fixed factorial design; any exception is stated
explicitly.

\section{Debiased estimation of effect-modifier coefficients}\label{sec:debias}

\subsection{Foldwise global centering of the baseline error}\label{sec:global}
Recall the assignment multiplier $q_S$ from Section~\ref{sec:setup}. Under
Assumption~\ref{as:rct},
\[
\kappa_2:=\E\{q_S(A)^2\mid X\}
=2^{-2(K-1)}\sum_{a\in\cA}p_a^{-1}.
\]
This is a finite constant under the fixed positive allocation probabilities in
Assumption~\ref{as:rct}, and it does not depend on $n$ or $S$, because
$\phi_S(a)^2=1$. Fix a nonrandom
baseline $g$; it may vary along the asymptotic sequence, but we suppress that index. For fold
$\ell$, write
\[
\Delta_{g,\ell}(x):=\hat g^{(-\ell)}(x)-g(x),
\qquad
R_{iS}:=\hat\psi_{iS}-\psi^{g}_{iS}
=-\Delta_{g,\ell}(X_i)q_S(A_i),\quad i\in I_\ell.
\]

For each fold, define
\[
\mathcal F_{-\ell}:=\sigma\{(X_i,A_i,Y_i):i\notin I_\ell\},
\qquad
\mathcal G_\ell:=\sigma\{\mathcal F_{-\ell},X_1,\ldots,X_n,I_1,\ldots,I_L\}.
\]

\begin{lemma}[Foldwise conditional centering]\label{lem:global}
Suppose Assumption~\ref{as:rct} holds.  For every fold $\ell$, assume
$\hat g^{(-\ell)}$ is $\mathcal F_{-\ell}$-measurable.  Conditional on
$\mathcal G_\ell$, the variables $\{R_{iS}:i\in I_\ell\}$ are independent and satisfy
\[
\E(R_{iS}\mid\mathcal G_\ell)=0,
\qquad
\E(R_{iS}^2\mid\mathcal G_\ell)=\kappa_2\Delta_{g,\ell}(X_i)^2,
\quad i\in I_\ell.
\]
\end{lemma}
\begin{proof}
For $i\in I_\ell$, cross-fitting makes $\Delta_{g,\ell}$ measurable with respect to
$\mathcal F_{-\ell}$.  Under randomized assignment and independence across observations,
conditional on $\mathcal G_\ell$ the assignments $\{A_i:i\in I_\ell\}$ are independent with
$\PP(A_i=a\mid\mathcal G_\ell)=p_a$.  For nonempty $S$,
\[
\E\{q_S(A_i)\mid\mathcal G_\ell\}
=2^{-(K-1)}\sum_{a\in\cA}\phi_S(a)=0
\]
by \eqref{eq:zerosum}, whereas
$\E\{q_S(A_i)^2\mid\mathcal G_\ell\}=\kappa_2$.  Multiplying by the fixed quantity
$-\Delta_{g,\ell}(X_i)$ gives the two conditional moments.
\end{proof}

Lemma~\ref{lem:global} separates the two effects of estimating the baseline. Replacing $g$ by
the cross-fitted $\hat g^{(-\ell)}$ creates no conditional bias because $R_{iS}$ has conditional
mean zero within each validation fold. It does create additional random variation. To make this
variation negligible in the debiased estimator, condition (P3) below requires the weighted
$L_2$ size of $\Delta_{g,\ell}=\hat g^{(-\ell)}-g$ to vanish. This is a consistency requirement only: it
imposes neither a prescribed $n^{-1/4}$ rate nor a product-rate condition.

\subsection{Debiased estimator and pointwise inference}\label{sec:pointwise}
We study inference for one effect-modifier coefficient $\theta_{jS}^\ast$. The construction has
the standard debiased/desparsified Lasso architecture developed in
\citet{ZhangZhang2014}, \citet{vandeGeer2014}, and
\citet{JavanmardMontanari2014}: an initial regularized fit is followed by an approximate
orthogonalization and a one-step bias correction. We introduce the sample nodewise quantities
together with an explanation of their population counterparts. The population decompositions
given below then identify the leading term in the asymptotic distribution; hats distinguish
sample quantities from their population counterparts.

\paragraph{Initial sparse fit.}
Center the cross-fitted score and design to remove the projection intercept: with
$\bar\psi_S:=\mathbb P_n\hat\psi_S$ and $\bar Z:=\mathbb P_n Z$, let
$\tilde\psi_{iS}:=\hat\psi_{iS}-\bar\psi_S$ and $\tilde Z_i:=Z_i-\bar Z$. For each contrast
$S$, use the Lasso \citep{Tibshirani1996} to estimate its projection coefficients:
\[
\hat\theta_S\in\argmin_{\theta\in\R^p}
\left\{\frac{1}{2n}\|\tilde\psi_S-\tilde Z\theta\|_2^2
+\lambda\|\theta\|_1\right\},
\qquad \lambda\asymp\sqrt{\frac{\log(2pM)}{n}}.
\]
Here $M=|\cS|=2^K-1$ is the number of nonempty factorial contrasts, so $pM$ is the total
number of covariate-by-contrast cells.

\paragraph{Residualizing the target feature.}
At the population level, $\gamma_j$ is the coefficient vector in the least-squares linear
projection of $Z_j$ onto $Z_{-j}$, and $V_j=Z_j-Z_{-j}^{\top}\gamma_j$ is its residual.
The corresponding normalizing constant is
$\tau_j^2=\E(V_jZ_j)=\E(V_j^2)$, where the equality follows from the projection normal
equations. These population quantities are formally defined in
\eqref{eq:population-nodewise-decomposition} below; their sample counterparts are constructed
as follows.

Following the nodewise-Lasso implementation of \citet{vandeGeer2014}, isolate the contribution
of feature $j$ by regressing $\tilde Z_j$ on the remaining features:
\[
\hat\gamma_j\in\argmin_{\gamma\in\R^{p-1}}
\left\{\frac1n\|\tilde Z_j-\tilde Z_{-j}\gamma\|_2^2
+2\lambda_{j,n}\|\gamma\|_1\right\}.
\]
Its nodewise residual and target-coordinate normalizing constant are
\[
\hat V_{ij}:=\tilde Z_{ij}-\tilde Z_{i,-j}^\top\hat\gamma_j,
\qquad
\hat\tau_j^2:=\mathbb P_n(\hat V_j\tilde Z_j).
\]
Here $\hat V_j$ is the variation in the target feature that remains after removing its sparse
linear prediction from the other features. The scalar $\hat\tau_j^2$ is the empirical loading
of this residualized direction on $\tilde Z_j$. Dividing by $\hat\tau_j^2$ makes that loading
equal to one,
\[
\frac{\mathbb P_n(\hat V_j\tilde Z_j)}{\hat\tau_j^2}=1,
\]
and thereby calibrates the one-step correction to the target coordinate. This is not a
unit-variance standardization: because the nodewise Lasso gives only approximate sample
orthogonality, $\hat\tau_j^2$ need not equal $\mathbb P_n(\hat V_j^2)$ in finite samples.
The same nodewise regression is reused for every contrast $S$, because the design $Z$ does not
depend on the contrast. The Karush--Kuhn--Tucker conditions for this nodewise Lasso
\citep[Section~2.1.1]{vandeGeer2014} give the single property needed to control the Lasso
shrinkage bias:
\begin{equation}\label{eq:nodewise-kkt}
\left\|\mathbb P_n(\hat V_j\tilde Z_{-j})\right\|_\infty\le\lambda_{j,n}.
\end{equation}

\paragraph{Debiasing.}
The initial coefficient $\hat\theta_{jS}$ is corrected in the residualized direction $\hat V_j$:
\begin{equation}\label{eq:debiased}
\hat b_{jS}:=\hat\theta_{jS}
+\frac{\mathbb P_n\{\hat V_j(\tilde\psi_S-\tilde Z\hat\theta_S)\}}
{\hat\tau_j^2}.
\end{equation}
Thus $\hat\theta_S$ provides the initial sparse fit, while $\hat V_j$ supplies the correction
specific to feature $j$. This is the usual coordinatewise one-step correction in the
debiased-Lasso literature cited above.

\paragraph{Population quantities governing the limit.}
We define here the population quantities used by both the pointwise and simultaneous limits.
For every $S\in\cS$, write the projection coefficient
$\beta_S^\ast=(\alpha_S^\ast,\theta_S^{\ast\top})^\top$ from
\eqref{eq:projection-target}. Because $\E Z=0$, its intercept is
$\alpha_S^\ast=\E(\psi_S^g)=\E\{\tau_S(X)\}$. Define the oracle score-projection residual by
the population decomposition
\begin{equation}\label{eq:population-score-decomposition}
\psi_S^g=\alpha_S^\ast+Z^\top\theta_S^\ast+\varepsilon_S.
\end{equation}
For every $j\in[p]$, define the population nodewise projection and its residual by
\begin{equation}\label{eq:population-nodewise-decomposition}
\gamma_j:=\argmin_{\gamma\in\R^{p-1}}\E(Z_j-Z_{-j}^\top\gamma)^2,
\qquad
Z_j=Z_{-j}^\top\gamma_j+V_j,
\qquad
\tau_j^2:=\E(V_j^2).
\end{equation}
For observation $i$, let $\varepsilon_{iS}$ and $V_{ij}$ denote the corresponding realizations
of $\varepsilon_S$ and $V_j$.
The normal equations for these two projections, together with $\E Z=0$, give
\[
\E(\varepsilon_S)=0,
\qquad
\E(Z\varepsilon_S)=0,
\qquad
\E(V_j)=0,
\qquad
\E(Z_{-j}V_j)=0.
\]
These are population projection decompositions, not correctly specified linear models;
in particular, the residuals need not have conditional mean zero. Here $V_j$ is the part of
feature $j$ orthogonal to the other features. Consequently,
$\E(V_jZ_j)=\E(V_j^2)=\tau_j^2$, so the population normalizing constant is also the
nodewise-residual variance.

The pointwise result below concerns one prespecified target cell
$(j,S)\in[p]\times\cS$; the notation suppresses any deterministic dependence of this cell on
$n$. For that cell, define the variance scale
\begin{equation}\label{eq:pointwise-asymptotic-variance}
\sigma_{jS}^2:=\frac{\E(V_j^2\varepsilon_S^2)}{\tau_j^4}.
\end{equation}
Theorem~\ref{thm:clt} below shows that this is the asymptotic variance of the debiased
estimator.

We now state the conditions used to establish pointwise asymptotic linearity.

\begin{assumption}[Pointwise conditions]\label{as:reg}
The following conditions hold for the target cell $(j,S)$. All fixed constants below are
independent of $n$ and $p$, but may depend on the fixed factorial design.
\begin{enumerate}[leftmargin=1.8em,label=(P\arabic*)]
\item \emph{Sampling and folds.}
The number of folds $L$ is fixed, the partition is deterministic or independent of the data, and
$n_\ell/n\to\pi_\ell\in(0,1)$ for every $\ell$.  Each $\hat g^{(-\ell)}$ is
$\mathcal F_{-\ell}$-measurable. The treatment assignment is governed by the fixed probabilities
in Assumption~\ref{as:rct}; hence $\kappa_2<\infty$ requires no additional asymptotic condition.

\item \emph{Projection, decorrelation, and moments.}
$\Sigma_Z=\E(ZZ^\top)$ is positive definite, so
$\theta_S^\ast$ and $\gamma_j$ are unique. For some $\delta>0$,
$0<c_\tau<C_\tau<\infty$, $0<c_\sigma<C_\sigma<\infty$, and $C_m<\infty$,
\[
c_\tau\le\tau_j^2\le C_\tau,
\qquad
c_\sigma\le\sigma_{jS}^2\le C_\sigma,
\]
and
\[
\E|V_j\varepsilon_S|^{2+\delta}\le C_m,
\qquad
\E(\varepsilon_S^2)\le C_m,
\qquad
\E|V_jZ_j|^{1+\delta}\le C_m.
\]

\item \emph{Foldwise baseline consistency.}
With $\Delta_{g,\ell}=\hat g^{(-\ell)}-g$,
\begin{equation}\label{eq:weighted-baseline}
\begin{aligned}
\max_{1\le\ell\le L}\mathbb P_{n,\ell}\Delta_{g,\ell}^2&=o_p(1),\\
\max_{1\le\ell\le L}\mathbb P_{n,\ell}(\hat V_j^2\Delta_{g,\ell}^2)&=o_p(1).
\end{aligned}
\end{equation}

\item \emph{Remainders from the two sparse regressions.}
For deterministic sequences $r_{\theta,n},r_{\gamma,n},r_{\varepsilon,n},b_{Z,n}\ge0$,
the estimation errors and empirical correlations satisfy
\[
\|\hat\theta_S-\theta_S^\ast\|_1=O_p(r_{\theta,n}),
\qquad
\|\hat\gamma_j-\gamma_j\|_1=O_p(r_{\gamma,n}),
\]
\[
\left\|\mathbb P_n(\tilde Z_{-j}\varepsilon_S)\right\|_\infty
=O_p(r_{\varepsilon,n}),
\qquad
\left\|\mathbb P_n(\tilde Z_{-j}\tilde Z_j)\right\|_\infty
=O_p(b_{Z,n}),
\]
where $r_{\theta,n}$ and $r_{\gamma,n}$ measure the two coefficient-estimation errors, while
$r_{\varepsilon,n}$ and $b_{Z,n}$ measure the corresponding empirical correlations. Only the
following products are required to vanish:
\begin{equation}\label{eq:pointwise-rates}
\sqrt n\,\lambda_{j,n}r_{\theta,n}\to0,
\qquad
\sqrt n\,r_{\gamma,n}r_{\varepsilon,n}\to0,
\qquad
r_{\gamma,n}b_{Z,n}\to0.
\end{equation}
\end{enumerate}
\end{assumption}

\begin{theorem}[Pointwise asymptotic linearity and normality]\label{thm:clt}
Under Assumptions~\ref{as:rct} and \ref{as:reg},
\begin{equation}\label{eq:pointwise-linearization}
\sqrt n(\hat b_{jS}-\theta_{jS}^\ast)
=\frac{1}{\tau_j^2\sqrt n}\sum_{i=1}^n V_{ij}\varepsilon_{iS}+o_p(1).
\end{equation}
Consequently,
\begin{equation}\label{eq:pointwise-clt}
\frac{\sqrt n(\hat b_{jS}-\theta_{jS}^\ast)}{\sigma_{jS}}
\ \Rightarrow\ N(0,1).
\end{equation}
The convergence in \eqref{eq:pointwise-clt} is under the joint sampling law.  If additionally
$\sigma_{jS}^2\to\sigma^2\in(0,\infty)$, then
$\sqrt n(\hat b_{jS}-\theta_{jS}^\ast)\Rightarrow N(0,\sigma^2)$.
\end{theorem}

The linear representation in \eqref{eq:pointwise-linearization} identifies
$V_j\varepsilon_S/\tau_j^2$ as the influence term and \eqref{eq:pointwise-asymptotic-variance}
as its variance.

\begin{proof}
Write $\delta_{\gamma,j}:=\hat\gamma_j-\gamma_j$,
$\bar V_j:=\mathbb P_nV_j$, and $\bar Z:=\mathbb P_nZ$. Since
$\tilde Z=Z-\bar Z$, each sample nodewise residual satisfies the exact identity
\begin{equation}\label{eq:vhat-expansion}
\begin{aligned}
\hat V_{ij}
&=\tilde Z_{ij}-\tilde Z_{i,-j}^{\top}\hat\gamma_j\\
&=(Z_{ij}-Z_{i,-j}^{\top}\gamma_j)
  -(\bar Z_j-\bar Z_{-j}^{\top}\gamma_j)
  -\tilde Z_{i,-j}^{\top}(\hat\gamma_j-\gamma_j)\\
&=V_{ij}-\bar V_j-\tilde Z_{i,-j}^{\top}\delta_{\gamma,j},
\qquad i=1,\ldots,n.
\end{aligned}
\end{equation}
We first show that its normalizing constant is consistent. Using
$\mathbb P_n\tilde Z_j=0$ and $\E(V_jZ_j)=\E(V_j^2)=\tau_j^2$ gives
\begin{align}
\hat\tau_j^2-\tau_j^2
={}&\{\mathbb P_n(V_jZ_j)-\E(V_jZ_j)\}
-\bar V_j\bar Z_j
-\delta_{\gamma,j}^{\top}\mathbb P_n(\tilde Z_{-j}\tilde Z_j).
\label{eq:tauhat-expansion}
\end{align}
Let $q=\min(1+\delta,2)>1$. Assumption~\ref{as:reg}(P2) gives
$\E|V_jZ_j|^q<\infty$, so the $L_q$ weak law yields
$\mathbb P_n(V_jZ_j)-\E(V_jZ_j)=o_p(1)$. Moreover,
$\E(V_j^2)=\tau_j^2\le C_\tau$ and $\E(Z_j^2)=1$, and hence Chebyshev's inequality gives
$\bar V_j=O_p(n^{-1/2})$ and $\bar Z_j=O_p(n^{-1/2})$. Finally, by H\"older's inequality
and Assumption~\ref{as:reg}(P4),
\[
\left|\delta_{\gamma,j}^{\top}
\mathbb P_n(\tilde Z_{-j}\tilde Z_j)\right|
\le
\|\delta_{\gamma,j}\|_1
\left\|\mathbb P_n(\tilde Z_{-j}\tilde Z_j)\right\|_\infty
=O_p(r_{\gamma,n}b_{Z,n})=o_p(1).
\]
Thus \eqref{eq:tauhat-expansion} shows that
\begin{equation}\label{eq:tauhat-consistency}
\hat\tau_j^2\to_p\tau_j^2,
\qquad
(\hat\tau_j^2)^{-1}=O_p(1),
\end{equation}
where the second conclusion uses $\tau_j^2\ge c_\tau$.

We next expand the numerator. Centering the design gives the exact identity
\[
\mathbb P_n\hat V_j
=\mathbb P_n\tilde Z_j-(\mathbb P_n\tilde Z_{-j})^\top\hat\gamma_j=0.
\]
Recall that $\psi_{iS}^g$ is the direct score defined in \eqref{eq:score}. By the definition of
$R_{iS}$ in Section~3.1 and the definition of the projection residual $\varepsilon_{iS}$,
\[
\hat\psi_{iS}=\psi_{iS}^g+R_{iS},
\qquad
\psi_{iS}^g=\alpha_S^\ast+Z_i^\top\theta_S^\ast+\varepsilon_{iS}.
\]
Thus the second equality is a decomposition of the score in \eqref{eq:score}, not a new
definition of that score. Since $\bar\psi_S-\bar Z^\top\theta_S^\ast$ is constant across
observations, $\mathbb P_n\hat V_j=0$ gives
\begin{align*}
&\mathbb P_n\{\hat V_j(\tilde\psi_S-\tilde Z^\top\theta_S^\ast)\}\\
&\quad=\mathbb P_n\!\left[\hat V_j
\{\hat\psi_S-Z^\top\theta_S^\ast-(\bar\psi_S-\bar Z^\top\theta_S^\ast)\}\right]\\
&\quad=\mathbb P_n\{\hat V_j(\hat\psi_S-Z^\top\theta_S^\ast)\}
=\mathbb P_n\{\hat V_j(\alpha_S^\ast+\varepsilon_S+R_S)\}\\
&\quad=\mathbb P_n\{\hat V_j(\varepsilon_S+R_S)\}.
\end{align*}
Substituting this identity into \eqref{eq:debiased} and using
$\mathbb P_n(\hat V_j\tilde Z_j)=\hat\tau_j^2$ gives the exact decomposition
\begin{equation}\label{eq:exact-decomposition}
\hat b_{jS}-\theta_{jS}^\ast
=\frac{\mathbb P_n\{\hat V_j(\varepsilon_S+R_S)\}}{\hat\tau_j^2}
-\frac{\sum_{k\ne j}\mathbb P_n(\hat V_j\tilde Z_k)
(\hat\theta_{kS}-\theta_{kS}^\ast)}{\hat\tau_j^2}.
\end{equation}
The nodewise KKT inequality \eqref{eq:nodewise-kkt} and (P4) give
\begin{align}
&\sqrt n\left|
\sum_{k\ne j}\mathbb P_n(\hat V_j\tilde Z_k)
(\hat\theta_{kS}-\theta_{kS}^\ast)
\right| \notag\\
&\qquad\le
\sqrt n\,\lambda_{j,n}\|\hat\theta_S-\theta_S^\ast\|_1
=O_p(\sqrt n\,\lambda_{j,n}r_{\theta,n})=o_p(1).
\label{eq:debiased-bias-bound}
\end{align}

For the oracle-score term, \eqref{eq:vhat-expansion} yields
\begin{align}
\sqrt n\,\mathbb P_n(\hat V_j\varepsilon_S)
-\frac1{\sqrt n}\sum_{i=1}^nV_{ij}\varepsilon_{iS}
=-\sqrt n\,\bar V_j\mathbb P_n\varepsilon_S
-\sqrt n\,\delta_{\gamma,j}^{\top}
\mathbb P_n(\tilde Z_{-j}\varepsilon_S).
\label{eq:oracle-replacement}
\end{align}
Because $\E\varepsilon_S=0$ and $\E(\varepsilon_S^2)\le C_m$,
$\mathbb P_n\varepsilon_S=O_p(n^{-1/2})$. The first term on the right of
\eqref{eq:oracle-replacement} is therefore $O_p(n^{-1/2})$. The second has absolute value at
most
\[
\sqrt n\,\|\delta_{\gamma,j}\|_1
\left\|\mathbb P_n(\tilde Z_{-j}\varepsilon_S)\right\|_\infty
=O_p(\sqrt n\,r_{\gamma,n}r_{\varepsilon,n})=o_p(1)
\]
by (P4). Consequently,
\begin{equation}\label{eq:oracle-replacement-result}
\sqrt n\,\mathbb P_n(\hat V_j\varepsilon_S)
=\frac1{\sqrt n}\sum_{i=1}^nV_{ij}\varepsilon_{iS}+o_p(1).
\end{equation}

It remains to control the baseline term. For each fold, define
\[
B_{\ell,n}:=\frac1{\sqrt n}\sum_{i\in I_\ell}\hat V_{ij}R_{iS}.
\]
The vector $\hat V_j$ is a function only of $Z_1,\ldots,Z_n$ and is therefore
$\mathcal G_\ell$-measurable. Lemma~\ref{lem:global} then gives
\begin{equation}\label{eq:baseline-fold-variance}
\E(B_{\ell,n}\mid\mathcal G_\ell)=0,
\qquad
\Var(B_{\ell,n}\mid\mathcal G_\ell)
=\kappa_2\frac{n_\ell}{n}
\mathbb P_{n,\ell}(\hat V_j^2\Delta_{g,\ell}^2)=o_p(1),
\end{equation}
where the last equality follows from (P1) and the second condition in (P3). To pass from the
conditional variance to unconditional convergence, fix $\epsilon,\eta>0$. Conditional
Chebyshev's inequality implies
\[
\PP(|B_{\ell,n}|>\epsilon)
\le
\PP\{\Var(B_{\ell,n}\mid\mathcal G_\ell)>\eta\}
+\frac{\eta}{\epsilon^2}.
\]
The first term tends to zero for every fixed $\eta$ by
\eqref{eq:baseline-fold-variance}; letting $\eta\downarrow0$ gives
$B_{\ell,n}=o_p(1)$. Since $L$ is fixed,
\begin{equation}\label{eq:baseline-remainder-result}
\sqrt n\,\mathbb P_n(\hat V_jR_S)
=\sum_{\ell=1}^LB_{\ell,n}=o_p(1).
\end{equation}

Combining \eqref{eq:exact-decomposition}, \eqref{eq:tauhat-consistency},
\eqref{eq:debiased-bias-bound}, \eqref{eq:oracle-replacement-result}, and
\eqref{eq:baseline-remainder-result} yields
\[
\sqrt n(\hat b_{jS}-\theta_{jS}^\ast)
=\frac1{\hat\tau_j^2}
\left\{\frac1{\sqrt n}\sum_{i=1}^nV_{ij}\varepsilon_{iS}+o_p(1)\right\}.
\]
Because $V_j$ is a linear combination of $Z$, the projection normal equations give
$\E(V_j\varepsilon_S)=0$. Hence the leading sum is $O_p(1)$ because its variance is
$\E(V_j^2\varepsilon_S^2)=\tau_j^4\sigma_{jS}^2=O(1)$. Slutsky's theorem and
\eqref{eq:tauhat-consistency} therefore prove \eqref{eq:pointwise-linearization}.

Finally, let $U_{ni}:=V_{ij}\varepsilon_{iS}$ and
$\nu_n^2:=\E(U_{ni}^2)=\tau_j^4\sigma_{jS}^2$. The projection normal equations imply
$\E U_{ni}=0$. Assumption~\ref{as:reg}(P2) bounds $\nu_n^2$ away from zero and infinity and
gives $\E|U_{ni}|^{2+\delta}\le C_m$. Hence the Lyapunov ratio satisfies
\[
\frac{\sum_{i=1}^n\E|U_{ni}|^{2+\delta}}
{(n\nu_n^2)^{1+\delta/2}}
=\frac{\E|U_{n1}|^{2+\delta}}
{n^{\delta/2}\nu_n^{2+\delta}}
\longrightarrow0.
\]
Lyapunov's central limit theorem for the row-wise i.i.d.\ array gives
$(\nu_n\sqrt n)^{-1}\sum_iU_{ni}\Rightarrow N(0,1)$. Since
$\nu_n=\tau_j^2\sigma_{jS}$, combining this limit with
\eqref{eq:pointwise-linearization} proves \eqref{eq:pointwise-clt}. If
$\sigma_{jS}^2\to\sigma^2\in(0,\infty)$, one further application of Slutsky's theorem gives
$\sqrt n(\hat b_{jS}-\theta_{jS}^\ast)\Rightarrow N(0,\sigma^2)$.
\end{proof}

\paragraph{Interpreting the rate conditions.}
Theorem~\ref{thm:clt} separates the baseline error from the two sparse-regression errors.
Condition (P3) asks only that the unweighted and weighted foldwise baseline errors vanish; it
imposes no fixed convergence rate. Condition (P4), by contrast, requires products of the initial-Lasso and
nodewise-Lasso errors to make the numerator remainders $o_p(n^{-1/2})$ and the estimated
normalization consistent.
To make these product restrictions concrete, suppose the bounds in (P4) take the familiar
sparse-regression forms
\[
\begin{aligned}
\lambda_{j,n}&\asymp\sqrt{\frac{\log(2p)}{n}},
&r_{\theta,n}&\asymp s_S\sqrt{\frac{\log(2pM)}{n}},
&r_{\gamma,n}&\asymp s_j\sqrt{\frac{\log(2p)}{n}},\\
r_{\varepsilon,n}&\asymp\sqrt{\frac{\log(2p)}{n}},
&b_{Z,n}&=O(1).
\end{aligned}
\]
Here $s_S=\|\theta_S^\ast\|_0$ and $s_j=\|\gamma_j\|_0$. Substitution into
\eqref{eq:pointwise-rates} reduces the required product rates to
\[
\frac{s_S\sqrt{\log(2p)\,\log(2pM)}}{\sqrt n}\to0,
\qquad
\frac{s_j\log(2p)}{\sqrt n}\to0.
\]
Thus $s_S\sqrt{\log(2p)\,\log(2pM)}=o(\sqrt n)$ controls the initial fit for contrast $S$,
while $s_j\log(2p)=o(\sqrt n)$ controls the nodewise regression for feature $j$. These are
sufficient numerical rates, not additional conclusions of Theorem~\ref{thm:clt}.

It remains to verify that using the estimated score as the Lasso response does not change the
initial-fit rate. Section~\ref{sec:primitive-pointwise} makes this comparison explicit. With
$\Delta_{g,Z,n}^2:=\max_{\ell,k}\mathbb P_{n,\ell}(\tilde Z_k^2\Delta_{g,\ell}^2)$ and
$B_{g,Z,n}:=\max_{\ell,i\in I_\ell,k}|\tilde Z_{ik}\Delta_{g,\ell}(X_i)|$, a foldwise conditional
Bernstein argument gives
\[
\left\|\mathbb P_n(\tilde ZR_S)\right\|_\infty
=O_p\!\left\{
\Delta_{g,Z,n}\sqrt{\frac{\log(2p)}{n}}
+B_{g,Z,n}\frac{\log(2p)}{n}
\right\}.
\]
The primitive conditions in Section~\ref{sec:primitive-pointwise} make the right-hand side
\[
o_p\!\left\{\sqrt{\frac{\log(2p)}{n}}\right\}.
\]
Since $\lambda\asymp\sqrt{\log(2pM)/n}$, this is also $o_p(\lambda)$. This penalty-scale
requirement differs from (P3): (P3)
controls one final influence-function remainder for the target $j$, whereas the score-Lasso
bound must hold over all $p$ covariate coordinates. Under the sub-Gaussian design conditions
of Section~\ref{sec:primitive-pointwise}, one simple sufficient route is
$\max_{\ell\le L}\|\hat g^{(-\ell)}-g\|_\infty=o_p(1)$ together with
$\log(2p)\log(2np)/n\to0$. The high-level formulation in Theorem~\ref{thm:clt} allows other
conditions whenever they deliver the same nuisance bounds.

\paragraph{Variance estimation and feasible inference.}
To make Theorem~\ref{thm:clt} operational, we next estimate the unknown asymptotic variance
$\sigma_{jS}^2$. We construct a sandwich estimator from the nodewise residuals and the fitted
projection residuals. The additional stability conditions below ensure that replacing the
population residuals by their estimated counterparts does not affect the leading second
moment. Consistency of the resulting estimator then yields a feasible pointwise confidence
interval.

Specifically, define
\[
\hat\varepsilon_{iS}:=\tilde\psi_{iS}-\tilde Z_i^\top\hat\theta_S
=\hat\psi_{iS}-\hat b_{0S}-Z_i^\top\hat\theta_S,
\qquad
\hat b_{0S}:=\bar\psi_S-\bar Z^\top\hat\theta_S.
\]

\begin{assumption}[Variance-estimation stability]\label{as:variance}
\[
\mathbb P_n\!\left[\hat V_j^2
\{\tilde Z^\top(\hat\theta_S-\theta_S^\ast)\}^2\right]=o_p(1),
\qquad
\mathbb P_n\{(\hat V_j-V_j)^2\varepsilon_S^2\}=o_p(1).
\]
\end{assumption}

\begin{corollary}[Consistent sandwich variance and valid pointwise interval]\label{cor:pointwise-ci}
Under Assumptions~\ref{as:rct}, \ref{as:reg}, and \ref{as:variance},
\begin{equation}\label{eq:var}
\hat\sigma^2_{jS}
:=\frac{\mathbb P_n(\hat V_j^2\hat\varepsilon_S^2)}{(\hat\tau_j^2)^2}
\quad\text{satisfies}\quad
\frac{\hat\sigma^2_{jS}}{\sigma^2_{jS}}\to_p1.
\end{equation}
Consequently,
\[
\frac{\sqrt n(\hat b_{jS}-\theta_{jS}^\ast)}{\hat\sigma_{jS}}
\Rightarrow N(0,1),
\]
and, for every fixed $\alpha\in(0,1)$,
\[
\PP\!\left\{
\theta_{jS}^\ast\in
\left[\hat b_{jS}\pm z_{1-\alpha/2}\frac{\hat\sigma_{jS}}{\sqrt n}\right]
\right\}\to1-\alpha.
\]
\end{corollary}

\begin{proof}
Write $\delta_{\theta,S}:=\hat\theta_S-\theta_S^\ast$,
$\bar\varepsilon_S:=\mathbb P_n\varepsilon_S$, and
$\bar R_S:=\mathbb P_nR_S$. The score decomposition used in the proof of
Theorem~\ref{thm:clt} gives
\begin{equation}\label{eq:variance-residual-decomposition}
\hat\varepsilon_{iS}
=\varepsilon_{iS}+D_{iS},
\qquad
D_{iS}:=-\bar\varepsilon_S+(R_{iS}-\bar R_S)
-\tilde Z_i^\top\delta_{\theta,S}.
\end{equation}

We first record the second-moment bound for the nodewise residual. Comparing the
nodewise-Lasso objective at $\hat\gamma_j$ with its value at zero gives
\[
\mathbb P_n\hat V_j^2+2\lambda_{j,n}\|\hat\gamma_j\|_1
\le \mathbb P_n\tilde Z_j^2
\le \mathbb P_nZ_j^2.
\]
Because $\E Z_j^2=1$, Markov's inequality yields
$\mathbb P_nZ_j^2=O_p(1)$ and hence
\begin{equation}\label{eq:vhat-second-moment}
\mathbb P_n\hat V_j^2=O_p(1).
\end{equation}

Next, $q_S$ is bounded on $\cA$ because the factorial design and its positive allocation
probabilities are fixed. For $i\in I_\ell$, $R_{iS}=-\Delta_{g,\ell}(X_i)q_S(A_i)$, so
Assumption~\ref{as:reg}(P3) and the fixed number of folds imply
\begin{equation}\label{eq:variance-baseline-bounds}
\mathbb P_nR_S^2=o_p(1),
\qquad
\mathbb P_n(\hat V_j^2R_S^2)=o_p(1).
\end{equation}
In particular, $\bar R_S^2\le\mathbb P_nR_S^2=o_p(1)$. Moreover,
$\bar\varepsilon_S=O_p(n^{-1/2})$ by Assumption~\ref{as:reg}(P2). Using
\eqref{eq:vhat-second-moment}, \eqref{eq:variance-baseline-bounds}, and the first condition
of Assumption~\ref{as:variance} in \eqref{eq:variance-residual-decomposition} gives
\begin{align}
\mathbb P_n(\hat V_j^2D_S^2)
\le{}&3\bar\varepsilon_S^2\mathbb P_n\hat V_j^2
+6\mathbb P_n(\hat V_j^2R_S^2)
+6\bar R_S^2\mathbb P_n\hat V_j^2\notag\\
&+3\mathbb P_n\!\left[
\hat V_j^2(\tilde Z^\top\delta_{\theta,S})^2\right]
=o_p(1).
\label{eq:variance-residual-error}
\end{align}

Assumption~\ref{as:reg}(P2) and the row-wise $L_{1+\delta/2}$ weak law give
\begin{equation}\label{eq:oracle-second-moment}
\mathbb P_n(V_j^2\varepsilon_S^2)
=\E(V_j^2\varepsilon_S^2)+o_p(1)
=\tau_j^4\sigma_{jS}^2+o_p(1)=O_p(1).
\end{equation}
The second condition of Assumption~\ref{as:variance}, together with Cauchy--Schwarz and
\eqref{eq:oracle-second-moment}, therefore yields
\begin{align}
&\left|\mathbb P_n\{(\hat V_j^2-V_j^2)\varepsilon_S^2\}\right|\\
&\quad\le
\{\mathbb P_n((\hat V_j-V_j)^2\varepsilon_S^2)\}^{1/2}
\{\mathbb P_n((\hat V_j+V_j)^2\varepsilon_S^2)\}^{1/2}
=o_p(1),
\end{align}
where the second factor is $O_p(1)$ because
$(\hat V_j+V_j)^2\le2(\hat V_j-V_j)^2+8V_j^2$. Hence
\begin{equation}\label{eq:vhat-oracle-second-moment}
\mathbb P_n(\hat V_j^2\varepsilon_S^2)
=\tau_j^4\sigma_{jS}^2+o_p(1)=O_p(1).
\end{equation}

Finally, another application of Cauchy--Schwarz gives
\begin{align*}
&\left|\mathbb P_n\{\hat V_j^2(\hat\varepsilon_S^2-\varepsilon_S^2)\}\right|\\
&\quad\le
\{\mathbb P_n(\hat V_j^2D_S^2)\}^{1/2}
\{\mathbb P_n[\hat V_j^2(\hat\varepsilon_S+\varepsilon_S)^2]\}^{1/2}
=o_p(1).
\end{align*}
Indeed, the second factor is $O_p(1)$ by
$\hat\varepsilon_S+\varepsilon_S=2\varepsilon_S+D_S$,
\eqref{eq:variance-residual-error}, and \eqref{eq:vhat-oracle-second-moment}. Consequently,
\[
\mathbb P_n(\hat V_j^2\hat\varepsilon_S^2)
=\tau_j^4\sigma_{jS}^2+o_p(1).
\]
Combining this result with \eqref{eq:tauhat-consistency} and
$\sigma_{jS}^2\ge c_\sigma$ proves
$\hat\sigma_{jS}^2/\sigma_{jS}^2\to_p1$. The studentized limit and the coverage statement
then follow from Theorem~\ref{thm:clt} and Slutsky's theorem.
\end{proof}

\begin{remark}[Balanced assignment and the role of the baseline]\label{rem:bal}
For $p_a=2^{-K}$,
$\psi_S^{g}=2\{Y-g(X)\}\phi_S(A)$ and $\phi_S(A)^2=1$, so the conditional raw-score
second moment
\[
\rho_g^2(X):=\E\{(\psi_S^{g})^2\mid X\}
=4\E[\{Y-g(X)\}^2\mid X]
\]
is exactly independent of $S$.  The projection-residual variance is not generally $S$-invariant:
with $m_S(X)=\alpha_S^\ast+Z^\top\theta_S^\ast$,
\[
\E(\varepsilon_S^2\mid X)
=\rho_g^2(X)-\tau_S(X)^2+\{\tau_S(X)-m_S(X)\}^2.
\]
Thus any claim of equal variances across contrasts requires an additional restriction on the
contrast means and approximation errors.  Because $m_S$ and $\tau_S$ do not depend on the
chosen baseline, only the first term on the right depends on $g$, and by
Proposition~\ref{prop:gopt} it is minimized \emph{pointwise in $x$} at
$g=g_{\mathrm{opt}}$. Since $V_j$ is a measurable function of $X$ (a linear combination of the
coordinates of $Z(X)$), the tower property gives
$\E(V_j^2\varepsilon_S^2)=\E\{V_j^2\,\E(\varepsilon_S^2\mid X)\}$, the iterated expectation
being justified by Fubini--Tonelli for the nonnegative integrand $V_j^2\varepsilon_S^2$.
Pointwise minimization under the nonnegative weight $V_j^2$ therefore shows that
$g_{\mathrm{opt}}$ also minimizes $\E(V_j^2\varepsilon_S^2)$ within the class of
baselines; validity itself does not require this choice.
\end{remark}

\section{Simultaneous inference}\label{sec:simul}

Section~\ref{sec:debias} developed inference for one prespecified effect-modifier cell. We now
extend that analysis to the full covariate-by-contrast grid. The factorial design and its
$M=2^K-1$ contrasts are fixed, while the grid size $pM$ may grow through the covariate
dimension $p=p_n$. The main challenge is to control the studentized errors uniformly while accounting
for their dependence across cells. We first establish a uniform linearization and an
infeasible Gaussian approximation. We then make that approximation operational by introducing
a multiplier bootstrap based on the estimated influence coordinates, which yields simultaneous
confidence bands and strong familywise error control through a Romano--Wolf step-down
procedure \citep{RomanoWolf2005}. We explain the role of each result as it is introduced; formal verification of a
checkable sufficient-condition route is deferred to
Section~\ref{sec:supp-uniform-sufficient} of the supplementary material.

\subsection{Studentized statistics}
\label{sec:simul-statistics}

Let
\[
\cH_n=[p]\times\cS,
\qquad d_n=|\cH_n|=pM,
\qquad \ell_n=\log(2nd_n).
\]
For $h=(j,S)\in\cH_n$, define the unstandardized and standardized influence coordinates
\[
u_{i,h}:=V_{ij}\varepsilon_{iS},
\qquad
\nu_h^2:=\E(u_{i,h}^2)=\tau_j^4\sigma_{jS}^2,
\qquad
\xi_{i,h}:=u_{i,h}/\nu_h.
\]
Thus $\E\xi_{i,h}=0$ and $\E\xi_{i,h}^2=1$.  Let
\[
\xi_i:=(\xi_{i,h})_{h\in\cH_n},
\qquad
\Sigma_\xi:=\E(\xi_i\xi_i^\top)
=\bigl(\Sigma_{\xi,h,h'}\bigr)_{h,h'\in\cH_n},
\qquad
\Sigma_{\xi,h,h'}:=\E(\xi_{i,h}\xi_{i,h'}).
\]
Because every coordinate is standardized, $\Sigma_\xi$ is both the covariance matrix and
the correlation matrix of the oracle influence vector $\xi_i$.  For the estimated score, write
\[
\hat s_{i,h}:=\hat V_{ij}\hat\varepsilon_{iS},
\qquad
\hat u_{i,h}:=\hat s_{i,h}-\mathbb P_n\hat s_h,
\qquad
\hat\nu_h^2:=\mathbb P_n\hat u_h^2.
\]
We use the centered simultaneous standard-error estimator
\begin{equation}\label{eq:uniform-se}
\hat\sigma^{\mathrm c}_{jS}:=\frac{\hat\nu_{jS}}{\hat\tau_j^2}
\end{equation}
and the centered studentized estimation error and observable null statistic
\[
T_{jS}:=\frac{\sqrt n(\hat b_{jS}-\theta_{jS}^\ast)}
{\hat\sigma^{\mathrm c}_{jS}},
\qquad
T^0_{jS}:=\frac{\sqrt n\,\hat b_{jS}}{\hat\sigma^{\mathrm c}_{jS}}.
\]
For a true null $H_{jS}:\theta_{jS}^\ast=0$, $T^0_{jS}=T_{jS}$. On the event that some
$\hat\nu_{jS}=0$ or $\hat\tau_j^2=0$, define the corresponding $\hat b_{jS}$,
$\hat\sigma^{\mathrm c}_{jS}$, $T_{jS}$, and $T^0_{jS}$ to be zero. This is only a
finite-sample convention: Assumptions~\ref{as:unif-rem} and
\ref{as:studentization} imply that the event has probability tending to zero.
The centering in $\hat u_{i,h}$ is immaterial pointwise. It provides the coordinatewise
variance estimator used here and will later allow us to estimate dependence across the full
grid.

\subsection{Conditions for simultaneous inference}
\label{sec:simul-conditions}

The three conditions below are those needed for the infeasible Gaussian approximation.
Assumption~\ref{as:unif-rem} makes the nuisance and cross-fitted baseline remainders
negligible uniformly over the grid. Assumption~\ref{as:hd-array} supplies the moment and
growth conditions for the high-dimensional Gaussian approximation, while
Assumption~\ref{as:studentization} ensures that the estimated coordinatewise scales are
uniformly accurate. The additional correlation-estimation condition needed for feasible
bootstrap inference is introduced only after Theorem~\ref{thm:uniform-gaussian}. Because the
factorial design is fixed and every $p_a$ is positive,
$\max_{S\in\cS,a\in\cA}|q_S(a)|<\infty$ automatically.

\begin{assumption}[Uniform nuisance and baseline remainders]
\label{as:unif-rem}
There are fixed constants $0<c_\tau<C_\tau<\infty$,
$0<c_\nu<C_\nu<\infty$ such that the following conditions hold.
\begin{enumerate}[leftmargin=1.8em,label=(U\arabic*)]
\item The number of folds $L$ is fixed; the fold partition is deterministic or independent of
the data; $n_\ell/n\to\pi_\ell\in(0,1)$; and each $\hat g^{(-\ell)}$ is
$\mathcal F_{-\ell}$-measurable. For every $(j,S)\in\cH_n$, the population quantities are
defined by \eqref{eq:population-score-decomposition} and
\eqref{eq:population-nodewise-decomposition}, and all corresponding projections are unique.
Moreover,
\[
c_\tau\le\min_{j\le p}\tau_j^2\le\max_{j\le p}\tau_j^2\le C_\tau,
\]
\[
c_\nu\le\min_{(j,S)\in\cH_n}\nu_{jS}
\le\max_{(j,S)\in\cH_n}\nu_{jS}\le C_\nu.
\]

\item There are deterministic sequences
$r_{\theta,n},r_{\gamma,n},r_{\varepsilon,n},r_{V,n},
 r_{\bar\varepsilon,n},r_{Z,n},r_{\tau,n},b_{Z,n},\lambda_{\gamma,n}\ge0$
such that
\[
\max_{S\in\cS}\|\hat\theta_S-\theta_S^\ast\|_1=O_p(r_{\theta,n}),
\qquad
\max_{j\le p}\|\hat\gamma_j-\gamma_j\|_1=O_p(r_{\gamma,n}),
\]
\[
\max_{S\in\cS}\|\mathbb P_n(\tilde Z\varepsilon_S)\|_\infty
=O_p(r_{\varepsilon,n}),
\quad
\max_{j\le p}|\mathbb P_nV_j|=O_p(r_{V,n}),
\quad
\max_{S\in\cS}|\mathbb P_n\varepsilon_S|=O_p(r_{\bar\varepsilon,n}),
\]
\[
\max_{j\le p}|\mathbb P_nZ_j|=O_p(r_{Z,n}),
\quad
\max_{j\le p}|\mathbb P_n(V_jZ_j)-\tau_j^2|=O_p(r_{\tau,n}),
\quad
\max_{j,k\le p}|\mathbb P_n(\tilde Z_j\tilde Z_k)|=O_p(b_{Z,n}),
\]
and $\max_{j\le p}\lambda_{j,n}\le\lambda_{\gamma,n}$.

\item The deterministic remainder budget satisfies
\begin{equation}\label{eq:uniform-rate-budget}
\sqrt{\ell_n}\left\{
\sqrt n\lambda_{\gamma,n}r_{\theta,n}
+\sqrt nr_{\gamma,n}r_{\varepsilon,n}
+\sqrt nr_{V,n}r_{\bar\varepsilon,n}
\right\}\longrightarrow0,
\end{equation}
and
\begin{equation}\label{eq:uniform-denominator-budget}
r_{\tau,n}+r_{V,n}r_{Z,n}+r_{\gamma,n}b_{Z,n}\longrightarrow0.
\end{equation}

\item With $\Delta_{g,\ell}=\hat g^{(-\ell)}-g$, define
\[
\Delta_{g,V,n}^2
:=\max_{1\le\ell\le L}\max_{j\le p}
\mathbb P_{n,\ell}(\hat V_j^2\Delta_{g,\ell}^2),
\qquad
B_{g,V,n}
:=\max_{1\le\ell\le L}\max_{i\in I_\ell}\max_{j\le p}
|\hat V_{ij}\Delta_{g,\ell}(X_i)|.
\]
Then
\begin{equation}\label{eq:uniform-baseline-budget}
\ell_n\Delta_{g,V,n}=o_p(1),
\qquad
\frac{\ell_n^{3/2}B_{g,V,n}}{\sqrt n}=o_p(1).
\end{equation}
\end{enumerate}
\end{assumption}

\begin{assumption}[High-dimensional influence array]
\label{as:hd-array}
There is a deterministic sequence $B_n\ge1$ such that, for $k=1,2$,
\begin{equation}\label{eq:cck-moments}
\max_{h\in\cH_n}\E|\xi_{i,h}|^{2+k}\le B_n^k,
\qquad
\max_{h\in\cH_n}\E\exp(|\xi_{i,h}|/B_n)\le2,
\end{equation}
and
\begin{equation}\label{eq:cck-growth}
\frac{B_n^2\ell_n^7}{n}\longrightarrow0.
\end{equation}
\end{assumption}

\begin{assumption}[Uniform studentization]
\label{as:studentization}
Let
\[
\mathcal E_{\nu,n}:=\left\{\min_{h\in\cH_n}\hat\nu_h>0\right\}.
\]
The estimated standard deviations are positive with probability approaching one,
\begin{equation}\label{eq:studentization-positivity}
\PP(\mathcal E_{\nu,n})\longrightarrow1,
\end{equation}
and, on $\mathcal E_{\nu,n}$, their uniform relative error satisfies
\begin{equation}\label{eq:studentization-relative-rate}
\ell_n\max_{h\in\cH_n}
\left|\frac{\hat\nu_h}{\nu_h}-1\right|=o_p(1).
\end{equation}
\end{assumption}

\subsection{Main results}
\label{sec:simul-main}

\begin{theorem}[Uniform linearization and Gaussian approximation]
\label{thm:uniform-gaussian}
Suppose Assumptions~\ref{as:rct}, \ref{as:unif-rem}, \ref{as:hd-array}, and
\ref{as:studentization} hold. For each $n$, let
$G_n=(G_{n,h})_{h\in\cH_n}$ be centered Gaussian with covariance $\Sigma_\xi$. Then:
\begin{enumerate}[leftmargin=1.6em,label=(\roman*)]
\item \emph{Uniform studentized linearization:}
\begin{equation}\label{eq:uniform-linearization}
\max_{h\in\cH_n}
\left|T_h-\frac1{\sqrt n}\sum_{i=1}^n\xi_{i,h}\right|
=o_p(\ell_n^{-1/2}).
\end{equation}

\item \emph{Gaussian approximation:}
\begin{equation}\label{eq:uniform-gaussian}
\sup_{t\in\R}\left|
\PP\!\left(\max_{h\in\cH_n}|T_h|\le t\right)
-\PP\!\left(\max_{h\in\cH_n}|G_{n,h}|\le t\right)
\right|\longrightarrow0.
\end{equation}
\end{enumerate}
\end{theorem}

\begin{proof}
For $h=(j,S)\in\cH_n$, define the unstandardized oracle sum and its standardized version by
\[
L_{n,h}:=\frac1{\sqrt n}\sum_{i=1}^n u_{i,h},
\qquad
S_{n,h}:=\frac{L_{n,h}}{\nu_h}
=\frac1{\sqrt n}\sum_{i=1}^n\xi_{i,h}.
\]
Let
\[
S_n:=(S_{n,h})_{h\in\cH_n},
\qquad
T:=(T_h)_{h\in\cH_n}
\]
denote the corresponding oracle and studentized vectors.
We first derive a uniform expansion of the numerator of $T_h$, then control its
studentization, and finally transfer the Gaussian approximation from $S_n$ to $T$.

\emph{Step 1: uniform control of the nodewise normalizers.}
Write
$\delta_{\gamma,j}:=\hat\gamma_j-\gamma_j$,
$\bar V_j:=\mathbb P_nV_j$, and $\bar Z:=\mathbb P_nZ$.
The algebra in \eqref{eq:tauhat-expansion}, applied to every $j\le p$, gives
\[
\hat\tau_j^2-\tau_j^2
=\{\mathbb P_n(V_jZ_j)-\tau_j^2\}
-\bar V_j\bar Z_j
-\delta_{\gamma,j}^{\top}\mathbb P_n(\tilde Z_{-j}\tilde Z_j).
\]
Hence Assumption~\ref{as:unif-rem}(U2) and H\"older's inequality imply
\begin{align}
\max_{j\le p}|\hat\tau_j^2-\tau_j^2|
&\le
\max_{j\le p}|\mathbb P_n(V_jZ_j)-\tau_j^2|
+\max_{j\le p}|\bar V_j|\max_{j\le p}|\bar Z_j| \notag\\
&\quad+
\max_{j\le p}\|\delta_{\gamma,j}\|_1
\max_{j,k\le p}|\mathbb P_n(\tilde Z_j\tilde Z_k)| \notag\\
&=O_p(r_{\tau,n}+r_{V,n}r_{Z,n}+r_{\gamma,n}b_{Z,n})
=o_p(1),
\label{eq:uniform-tau-consistency}
\end{align}
where the last equality is (U3). Indeed, on the event
\[
\max_{j\le p}|\hat\tau_j^2-\tau_j^2|<c_\tau/2,
\]
the lower bound $\min_{j\le p}\tau_j^2\ge c_\tau$ gives
\[
\min_{j\le p}\hat\tau_j^2
\ge \min_{j\le p}\tau_j^2
-\max_{j\le p}|\hat\tau_j^2-\tau_j^2|
\ge c_\tau/2.
\]
The probability of the first event tends to one by
\eqref{eq:uniform-tau-consistency}; hence
\begin{equation}
\PP\!\left(\min_{j\le p}\hat\tau_j^2\ge c_\tau/2\right)\longrightarrow1.
\label{eq:uniform-tau-positive}
\end{equation}

\emph{Step 2: uniform expansion of the debiased numerator.}
Let $\delta_{\theta,S}:=\hat\theta_S-\theta_S^\ast$.
The exact decompositions \eqref{eq:exact-decomposition} and
\eqref{eq:oracle-replacement} hold for every $(j,S)$. Therefore
\begin{equation}
\sqrt n\,\hat\tau_j^2(\hat b_{jS}-\theta_{jS}^\ast)
=L_{n,h}+\mathcal R_{n,h},
\label{eq:uniform-exact-decomposition}
\end{equation}
where
\begin{align*}
\mathcal R_{n,h}
:={}&-\sqrt n\sum_{k\ne j}
\mathbb P_n(\hat V_j\tilde Z_k)\delta_{\theta,kS}
-\sqrt n\,\bar V_j\mathbb P_n\varepsilon_S\\
&-\sqrt n\,\delta_{\gamma,j}^{\top}
\mathbb P_n(\tilde Z_{-j}\varepsilon_S)
+\sqrt n\,\mathbb P_n(\hat V_jR_S).
\end{align*}
The nodewise KKT inequality \eqref{eq:nodewise-kkt} and (U2) yield, uniformly over
$h=(j,S)$,
\begin{align}
&\max_{h\in\cH_n}
\left|\mathcal R_{n,h}-\sqrt n\,\mathbb P_n(\hat V_jR_S)\right| \notag\\
&\quad\le
\sqrt n\max_{j\le p}\lambda_{j,n}
\max_{S\in\cS}\|\delta_{\theta,S}\|_1
+\sqrt n\max_{j\le p}|\bar V_j|
\max_{S\in\cS}|\mathbb P_n\varepsilon_S| \notag\\
&\qquad+
\sqrt n\max_{j\le p}\|\delta_{\gamma,j}\|_1
\max_{S\in\cS}\|\mathbb P_n(\tilde Z\varepsilon_S)\|_\infty \notag\\
&\quad=O_p\!\left(
\sqrt n\lambda_{\gamma,n}r_{\theta,n}
+\sqrt nr_{V,n}r_{\bar\varepsilon,n}
+\sqrt nr_{\gamma,n}r_{\varepsilon,n}
\right)
=o_p(\ell_n^{-1/2}),
\label{eq:uniform-nonbaseline-remainder}
\end{align}
where the last equality follows from \eqref{eq:uniform-rate-budget}.

It remains to establish the same order for the baseline term. For every fold $\ell$ and
$h=(j,S)$, set
\[
C_{\ell,h}:=\frac1{\sqrt n}\sum_{i\in I_\ell}\hat V_{ij}R_{iS},
\]
and define the $\mathcal G_\ell$-measurable quantities
\[
\Delta_{\ell,n}^2:=\max_{j\le p}
\mathbb P_{n,\ell}(\hat V_j^2\Delta_{g,\ell}^2),
\qquad
H_{\ell,n}:=\max_{i\in I_\ell}\max_{j\le p}
|\hat V_{ij}\Delta_{g,\ell}(X_i)|.
\]
Lemma~\ref{lem:global} shows that, conditional on $\mathcal G_\ell$, the summands in
$C_{\ell,h}$ are independent and centered. To verify the concentration bound explicitly, set
\[
X_{i,\ell,h}:=\frac{\hat V_{ij}R_{iS}}{\sqrt n},
\qquad i\in I_\ell,\quad h=(j,S),
\]
so that $C_{\ell,h}=\sum_{i\in I_\ell}X_{i,\ell,h}$. Because $\hat V_j$ and
$\Delta_{g,\ell}(X_i)$ are $\mathcal G_\ell$-measurable, Lemma~\ref{lem:global} gives
\[
\E(X_{i,\ell,h}\mid\mathcal G_\ell)=0
\]
and
\begin{align*}
\sum_{i\in I_\ell}
\E(X_{i,\ell,h}^2\mid\mathcal G_\ell)
&=\frac{\kappa_2}{n}
\sum_{i\in I_\ell}\hat V_{ij}^2\Delta_{g,\ell}(X_i)^2\\
&=\kappa_2\frac{n_\ell}{n}
\mathbb P_{n,\ell}(\hat V_j^2\Delta_{g,\ell}^2)
\le \kappa_2\Delta_{\ell,n}^2.
\end{align*}
Moreover, since the factorial design and the allocation probabilities are fixed,
\[
Q:=\max_{S\in\cS}\max_{a\in\cA}|q_S(a)|
=2^{-(K-1)}\max_{a\in\cA}p_a^{-1}<\infty.
\]
The identity $R_{iS}=-\Delta_{g,\ell}(X_i)q_S(A_i)$ therefore implies the conditional envelope
\[
|X_{i,\ell,h}|
\le \frac{QH_{\ell,n}}{\sqrt n}
\qquad\text{for every }i\in I_\ell\text{ and }h\in\cH_n.
\]
Applying the bounded-variable Bernstein inequality conditionally on $\mathcal G_\ell$
(equivalently, applying the ordinary inequality under the regular conditional law; see
\citealp[Theorem~2.10]{BoucheronLugosiMassart2013} for the bounded-variable inequality and
\citealp[Lemmas~2.2.9--2.2.10]{vanDerVaartWellner1996} for that inequality and its maximal
consequence) yields, for every $x>0$ and fixed
$h\in\cH_n$,
\begin{align*}
\PP\!\left[
|C_{\ell,h}|>
\sqrt{2\kappa_2}\,\Delta_{\ell,n}\sqrt{x}
+\frac{2QH_{\ell,n}x}{3\sqrt n}
\,\middle|\,\mathcal G_\ell\right]
\le 2e^{-x}.
\end{align*}
Taking $x=2\ell_n$ and then applying the union bound over the
$d_n=|\cH_n|$ cells gives a fixed design-dependent constant
$C=\max\{2\sqrt{\kappa_2},4Q/3\}<\infty$ such that
\begin{equation}
\PP\!\left[
\max_{h\in\cH_n}|C_{\ell,h}|>
C\left\{\Delta_{\ell,n}\sqrt{\ell_n}
+\frac{H_{\ell,n}\ell_n}{\sqrt n}\right\}
\,\middle|\,\mathcal G_\ell\right]
\le 2d_n e^{-2\ell_n}.
\label{eq:uniform-baseline-bernstein}
\end{equation}
Here independence across cells is not required for the union bound. In addition,
$\ell_n=\log(2nd_n)$ implies
$2d_ne^{-2\ell_n}=1/(2n^2d_n)\to0$. Finally, the fold partition gives the exact identity
\[
\sqrt n\,\mathbb P_n(\hat V_jR_S)
=\frac1{\sqrt n}\sum_{i=1}^n\hat V_{ij}R_{iS}
=\sum_{\ell=1}^LC_{\ell,h}.
\]
Because $L$ is fixed, summing the foldwise bounds and using
$\Delta_{g,V,n}=\max_\ell\Delta_{\ell,n}$ and
$B_{g,V,n}=\max_\ell H_{\ell,n}$ yields
\begin{align}
\max_{h\in\cH_n}
\left|\sqrt n\,\mathbb P_n(\hat V_jR_S)\right|
&=O_p\!\left(
\Delta_{g,V,n}\sqrt{\ell_n}
+\frac{B_{g,V,n}\ell_n}{\sqrt n}
\right) \notag\\
&=o_p(\ell_n^{-1/2})
\label{eq:uniform-baseline-remainder}
\end{align}
by \eqref{eq:uniform-baseline-budget}. Combining
\eqref{eq:uniform-nonbaseline-remainder} and
\eqref{eq:uniform-baseline-remainder} gives
\begin{equation}
\max_{h\in\cH_n}|\mathcal R_{n,h}|=o_p(\ell_n^{-1/2}).
\label{eq:uniform-numerator-remainder}
\end{equation}

\emph{Step 3: Gaussian approximation and size of the oracle maximum.}
The vectors $\xi_i=(\xi_{i,h})_{h\in\cH_n}$ are independent and centered, every coordinate
has variance one, and hence condition (M.1) of Proposition~2.1 in \citet{CCK2017} holds
with $b=1$. Assumption~\ref{as:hd-array} supplies conditions (M.2) and (E.1) of that
proposition. Applying it in dimension $d_n$ gives
\begin{equation}
\rho_n:=\sup_{A\in\mathcal R_{d_n}}
\left|\PP(S_n\in A)-\PP(G_n\in A)\right|
\le C\left(\frac{B_n^2\ell_n^7}{n}\right)^{1/6}=o(1),
\label{eq:oracle-hyperrectangle-clt}
\end{equation}
where $\mathcal R_{d_n}$ is the class of hyperrectangles in $\R^{d_n}$.
In particular, the Gaussian union bound and
\eqref{eq:oracle-hyperrectangle-clt} imply
\begin{equation}
\max_{h\in\cH_n}|S_{n,h}|=O_p(\sqrt{\ell_n}).
\label{eq:oracle-maximum-order}
\end{equation}

\emph{Step 4: uniform studentization.}
Let
\[
a_n:=\max_{h\in\cH_n}
\left|\frac{\hat\nu_h}{\nu_h}-1\right|.
\]
Assumption~\ref{as:studentization} gives $a_n=o_p(\ell_n^{-1})$ on
$\mathcal E_{\nu,n}$ and $\PP(\mathcal E_{\nu,n})\to1$. Define
\[
\Omega_n
:=\mathcal E_{\nu,n}
\cap\{a_n\le1/2\}
\cap\left\{\min_{j\le p}\hat\tau_j^2\ge c_\tau/2\right\}.
\]
Because $a_n=o_p(1)$ on $\mathcal E_{\nu,n}$ and
\eqref{eq:uniform-tau-positive} holds, $\PP(\Omega_n)\to1$. On $\Omega_n$, the lower bound
$\min_h\nu_h\ge c_\nu$ in (U1) gives
\[
\min_{h\in\cH_n}\hat\nu_h
\ge (1-a_n)\min_{h\in\cH_n}\nu_h
\ge c_\nu/2,
\]
and, for every $h\in\cH_n$,
\[
\left|\frac{\nu_h}{\hat\nu_h}-1\right|
=\frac{|1-\hat\nu_h/\nu_h|}{\hat\nu_h/\nu_h}
\le\frac{a_n}{1-a_n}
\le2a_n.
\]
Consequently,
\[
\max_{h\in\cH_n}\left|\frac{\nu_h}{\hat\nu_h}-1\right|
\le 2a_n.
\]
Both $\hat\tau_j^2$ and $\hat\nu_h$ are therefore positive on $\Omega_n$, so the finite-sample
zero convention is inactive and all statistics have their usual definitions. For
$h=(j,S)$, the factor $\hat\tau_j^2$ then cancels exactly:
\[
T_h
=\frac{\sqrt n(\hat b_{jS}-\theta_{jS}^\ast)}
{\hat\nu_h/\hat\tau_j^2}
=\frac{\sqrt n\,\hat\tau_j^2
(\hat b_{jS}-\theta_{jS}^\ast)}{\hat\nu_h}.
\]
Substituting \eqref{eq:uniform-exact-decomposition} and using
$L_{n,h}=\nu_hS_{n,h}$ gives, on $\Omega_n$, the exact identity
\[
T_h-S_{n,h}
=\frac{\mathcal R_{n,h}}{\hat\nu_h}
+\left(\frac{\nu_h}{\hat\nu_h}-1\right)S_{n,h}.
\]
Equations \eqref{eq:uniform-numerator-remainder} and
\eqref{eq:oracle-maximum-order} now imply
\[
\max_{h\in\cH_n}|T_h-S_{n,h}|
=o_p(\ell_n^{-1/2})
+o_p(\ell_n^{-1})O_p(\sqrt{\ell_n})
=o_p(\ell_n^{-1/2}),
\]
which proves part~(i).

\emph{Step 5: transfer to the Gaussian maximum.}
Let
$M_n^T:=\max_h|T_h|$, $M_n^S:=\max_h|S_{n,h}|$, and
$M_n^G:=\max_h|G_{n,h}|$. Write
\[
D_n:=\max_{h\in\cH_n}|T_h-S_{n,h}|.
\]
Part~(i) states that $\sqrt{\ell_n}D_n\to_p0$. Hence, by a deterministic
diagonal argument, there is a positive sequence
$\delta_n=o(\ell_n^{-1/2})$ such that
\[
\eta_n:=\PP(D_n>\delta_n)\longrightarrow0.
\]
Define $E_n:=\{D_n\le\delta_n\}$. On $E_n$, the reverse triangle inequality gives
\[
|M_n^T-M_n^S|
\le \max_{h\in\cH_n}\bigl||T_h|-|S_{n,h}|\bigr|
\le D_n
\le\delta_n.
\]
Consequently, for every $t\in\R$,
\[
\{M_n^S\le t-\delta_n\}\cap E_n
\subseteq \{M_n^T\le t\},
\qquad
\{M_n^T\le t\}\cap E_n
\subseteq \{M_n^S\le t+\delta_n\}.
\]
Taking probabilities and using $\PP(E_n^c)=\eta_n$ yields the distribution-function
sandwich
\begin{equation}
\PP(M_n^S\le t-\delta_n)-\eta_n
\le \PP(M_n^T\le t)
\le \PP(M_n^S\le t+\delta_n)+\eta_n.
\label{eq:maximum-cdf-sandwich}
\end{equation}

We next replace the oracle maximum by its Gaussian analogue. For $t\ge0$,
$\{M_n^S\le t\}=\{S_n\in[-t,t]^{d_n}\}$, so
\eqref{eq:oracle-hyperrectangle-clt} gives
\[
\sup_{t\in\R}|\PP(M_n^S\le t)-\PP(M_n^G\le t)|\le\rho_n=o(1).
\]
The same bound is automatic for $t<0$, since both maxima are nonnegative. Combining this
uniform Gaussian approximation with \eqref{eq:maximum-cdf-sandwich}, and writing
$F_n^T(t):=\PP(M_n^T\le t)$ and $F_n^G(t):=\PP(M_n^G\le t)$, gives
\begin{equation}
F_n^G(t-\delta_n)-\rho_n-\eta_n
\le F_n^T(t)
\le F_n^G(t+\delta_n)+\rho_n+\eta_n
\qquad (t\in\R).
\label{eq:maximum-gaussian-sandwich}
\end{equation}

It remains to show that shifting the Gaussian threshold by $\delta_n$ has an asymptotically
negligible effect. Write $M_n^G$ as the maximum of the $2d_n$-dimensional Gaussian vector
$(G_n^\top,-G_n^\top)^\top$. Every coordinate has variance one, so the Gaussian
anti-concentration inequality in Corollary~1 of \citet{CCKComparison2013} gives
\[
\alpha_n:=\sup_{t\in\R}\PP(|M_n^G-t|\le\delta_n)
\le C\delta_n\sqrt{1\vee\log(2d_n)}.
\]
For all sufficiently large $n$,
$1\vee\log(2d_n)\le\ell_n=\log(2nd_n)$; hence
\[
\alpha_n\le C\delta_n\sqrt{\ell_n}=o(1).
\]
Moreover, for every $t\in\R$,
\[
F_n^G(t+\delta_n)-F_n^G(t)\le\alpha_n,
\qquad
F_n^G(t)-F_n^G(t-\delta_n)\le\alpha_n.
\]
Subtracting $F_n^G(t)$ from the upper bound in
\eqref{eq:maximum-gaussian-sandwich} gives
\[
F_n^T(t)-F_n^G(t)
\le F_n^G(t+\delta_n)-F_n^G(t)+\rho_n+\eta_n
\le\alpha_n+\rho_n+\eta_n.
\]
Similarly, the lower bound in \eqref{eq:maximum-gaussian-sandwich} gives
\[
F_n^G(t)-F_n^T(t)
\le F_n^G(t)-F_n^G(t-\delta_n)+\rho_n+\eta_n
\le\alpha_n+\rho_n+\eta_n.
\]
Thus
\[
\sup_{t\in\R}|F_n^T(t)-F_n^G(t)|
\le \rho_n+\eta_n+\alpha_n\longrightarrow0.
\]
Equivalently,
\[
\sup_{t\in\R}|\PP(M_n^T\le t)-\PP(M_n^G\le t)|\longrightarrow0,
\]
which is part~(ii).
\end{proof}

\paragraph{Interpretation of Theorem~\ref{thm:uniform-gaussian}.}
Part (i) is the central extension of the pointwise result: it replaces the fixed-cell
$o_p(1)$ remainder in Theorem~\ref{thm:clt} by a representation that holds over all
$d_n=pM$ cells with error $o_p(\ell_n^{-1/2})$. This sharper scale allows the linearization
to survive the anti-concentration of a growing maximum. Part (ii) then approximates the
distribution of that maximum by a Gaussian vector without imposing a structural restriction
on its correlation matrix $\Sigma_\xi$. This Gaussian approximation is not yet operational,
however, because $\Sigma_\xi$ is unknown.

\paragraph{Feasible approximation by the multiplier bootstrap.}
Following the multiplier-bootstrap theory for high-dimensional maxima in
\citet{CCK2013,CCK2017}, we estimate the distribution of the Gaussian maximum using the
centered influence coordinates introduced in Section~\ref{sec:simul-statistics}.  Let
\[
\mathcal D_n
:=\sigma\!\left\{
(X_i,A_i,Y_i)_{i=1}^n,
I_1,\ldots,I_L,
(\hat g^{(-\ell)})_{\ell=1}^L,
(\hat\theta_S)_{S\in\cS},
(\hat\gamma_j)_{j=1}^p
\right\}.
\]
Thus $\mathcal D_n$ contains the observed sample, the fold partition, and all fitted quantities
used to construct the estimated influence coordinates.  On an extension of the probability
space, generate the multipliers and, on $\mathcal E_{\nu,n}$, define
\begin{equation}\label{eq:multiplier-bootstrap}
\begin{gathered}
(e_1,\ldots,e_n)\mid\mathcal D_n\sim N(0,I_n),\\
W_h^\flat:=\frac1{\sqrt n}\sum_{i=1}^n
e_i\frac{\hat u_{i,h}}{\hat\nu_h},
\qquad h\in\cH_n.
\end{gathered}
\end{equation}
On $\mathcal E_{\nu,n}^c$, set $W_h^\flat=0$ for every $h$; this arbitrary convention is
asymptotically irrelevant by \eqref{eq:studentization-positivity}.
For subsequent conditional statements, write
\[
\PP^\flat(\,\cdot\,):=\PP(\,\cdot\mid\mathcal D_n),
\qquad
\E^\flat(\,\cdot\,):=\E(\,\cdot\mid\mathcal D_n),
\qquad
\operatorname{Cov}^\flat(\,\cdot\,,\,\cdot\,)
:=\operatorname{Cov}(\,\cdot\,,\,\cdot\mid\mathcal D_n).
\]
For every nonempty
$H\subseteq\cH_n$ and $\alpha\in(0,1)$, let
\begin{equation}\label{eq:bootstrap-critical-value}
c_{1-\alpha}^\flat(H)
:=\inf\left\{t\in\R:
\PP^\flat\!\left(\max_{h\in H}|W_h^\flat|\le t\right)\ge1-\alpha\right\},
\qquad
c_{1-\alpha}^\flat:=c_{1-\alpha}^\flat(\cH_n).
\end{equation}
The following condition requires the conditional correlation matrix of
$W^\flat=(W_h^\flat)_{h\in\cH_n}$ to estimate its population counterpart uniformly.

\begin{assumption}[Uniform estimated-score correlation]
\label{as:score-cov}
On the event $\mathcal E_{\nu,n}$ from Assumption~\ref{as:studentization}, define
\[
\hat\Sigma_{h,h'}
:=\mathbb P_n\!\left(\frac{\hat u_h\hat u_{h'}}{\hat\nu_h\hat\nu_{h'}}\right),
\qquad
\Delta_{\Sigma,n}
:=\max_{h,h'\in\cH_n}|\hat\Sigma_{h,h'}-\Sigma_{\xi,h,h'}|.
\]
On this event, the uniform correlation error satisfies
\begin{equation}\label{eq:bootstrap-correlation-rate}
\ell_n^2\Delta_{\Sigma,n}=o_p(1).
\end{equation}
\end{assumption}

\begin{theorem}[Multiplier-bootstrap approximation and simultaneous confidence bands]
\label{thm:bootstrap-band}
Under the assumptions of Theorem~\ref{thm:uniform-gaussian} and
Assumption~\ref{as:score-cov}, the following conclusions hold:
\begin{enumerate}[leftmargin=1.6em,label=(\roman*)]
\item \emph{Conditional multiplier-bootstrap approximation:}
\begin{equation}\label{eq:uniform-bootstrap}
\sup_{t\in\R}\left|
\PP\!\left(\max_{h\in\cH_n}|T_h|\le t\right)
-\PP^\flat\!\left(\max_{h\in\cH_n}|W_h^\flat|\le t\right)
\right|\longrightarrow_p0.
\end{equation}

\item \emph{Simultaneous confidence band:} For every fixed $\alpha\in(0,1)$,
\begin{equation}\label{eq:simultaneous-coverage}
\PP\!\left\{
\theta_{jS}^\ast\in
\left[\hat b_{jS}\pm c_{1-\alpha}^\flat
\frac{\hat\sigma^{\mathrm c}_{jS}}{\sqrt n}\right]
\text{ for every }(j,S)\in\cH_n
\right\}\longrightarrow1-\alpha.
\end{equation}
\end{enumerate}
\end{theorem}

\begin{proof}
Let
\[
M_n^T:=\max_{h\in\cH_n}|T_h|,
\qquad
M_n^G:=\max_{h\in\cH_n}|G_{n,h}|,
\qquad
M_n^\flat:=\max_{h\in\cH_n}|W_h^\flat|.
\]
We first compare the conditional bootstrap law with the Gaussian law in
Theorem~\ref{thm:uniform-gaussian}, and then use deterministic Gaussian quantiles to handle
the data-dependent bootstrap critical value.

\emph{Step 1: conditional Gaussian comparison.}
On $\mathcal E_{\nu,n}$, equation~\eqref{eq:multiplier-bootstrap} shows that, conditional on
$\mathcal D_n$, $W^\flat$ is a centered Gaussian vector. Its conditional covariance satisfies
\begin{align*}
\operatorname{Cov}^\flat(W_h^\flat,W_{h'}^\flat)
&=\frac1n\sum_{i=1}^n
\frac{\hat u_{i,h}\hat u_{i,h'}}{\hat\nu_h\hat\nu_{h'}}
=\hat\Sigma_{h,h'},\\
\hat\Sigma_{h,h}
&=\frac{\mathbb P_n\hat u_h^2}{\hat\nu_h^2}=1.
\end{align*}
Since $G_n$ is centered Gaussian with covariance matrix $\Sigma_\xi$ and unit coordinate
variances, the largest entrywise discrepancy between the covariance matrices of $W^\flat$ and
$G_n$ is $\Delta_{\Sigma,n}$.

To treat absolute maxima, augment the two vectors as
\[
\widetilde W^\flat:=((W^\flat)^\top,-(W^\flat)^\top)^\top,
\qquad
\widetilde G_n:=(G_n^\top,-G_n^\top)^\top.
\]
Their coordinatewise maxima are $M_n^\flat$ and $M_n^G$, respectively, and the largest
entrywise difference between their covariance matrices is still
$\Delta_{\Sigma,n}$. On the event
$\mathcal E_{\nu,n}\cap\{0<\Delta_{\Sigma,n}\le1\}$, Theorem~2 of
\citet{CCKComparison2013}, applied in dimension $2d_n$, therefore gives
\begin{equation}
\kappa_n
:=\sup_{t\in\R}
\left|\PP^\flat(M_n^\flat\le t)-\PP(M_n^G\le t)\right|
\le
C\Delta_{\Sigma,n}^{1/3}
\left\{1\vee\log\left(\frac{2d_n}{\Delta_{\Sigma,n}}\right)\right\}^{2/3}.
\label{eq:bootstrap-gaussian-comparison}
\end{equation}
When $\Delta_{\Sigma,n}=0$, the two Gaussian laws coincide and we set the right-hand side to
zero. The event
$\mathcal E_{\nu,n}\cap\{\Delta_{\Sigma,n}>1\}$ has probability tending to zero by
Assumption~\ref{as:score-cov}; all statements involving $\Delta_{\Sigma,n}$ below are
understood on $\mathcal E_{\nu,n}$.

We verify explicitly that the right-hand side of
\eqref{eq:bootstrap-gaussian-comparison} vanishes. For
$0<\Delta_{\Sigma,n}\le1$,
\begin{align*}
&\Delta_{\Sigma,n}
\left\{1\vee\log\left(\frac{2d_n}{\Delta_{\Sigma,n}}\right)\right\}^{2}\\
&\qquad\le
C\left\{
\Delta_{\Sigma,n}\ell_n^2
+\Delta_{\Sigma,n}\log^2(1/\Delta_{\Sigma,n})
\right\}=o_p(1).
\end{align*}
Indeed, $\log(2d_n)\le\ell_n$, the first term is $o_p(1)$ by
\eqref{eq:bootstrap-correlation-rate}, and the second is $o_p(1)$ because
$\Delta_{\Sigma,n}\to_p0$ and $x\log^2(1/x)\to0$ as $x\downarrow0$. Together with
$\PP(\mathcal E_{\nu,n})\to1$, this proves
\begin{equation}
\kappa_n=o_p(1).
\label{eq:conditional-bootstrap-gaussian}
\end{equation}

Theorem~\ref{thm:uniform-gaussian}(ii) gives the nonrandom bound
\[
r_{G,n}:=\sup_{t\in\R}
\left|\PP(M_n^T\le t)-\PP(M_n^G\le t)\right|\longrightarrow0.
\]
Hence the triangle inequality and \eqref{eq:conditional-bootstrap-gaussian} yield
\[
\sup_{t\in\R}
\left|\PP(M_n^T\le t)-\PP^\flat(M_n^\flat\le t)\right|
\le r_{G,n}+\kappa_n\longrightarrow_p0,
\]
which proves part~(i).

\emph{Step 2: bootstrap quantile bracketing.}
Write
\[
F_n^G(t):=\PP(M_n^G\le t),
\qquad
F_n^\flat(t):=\PP^\flat(M_n^\flat\le t),
\]
and, for $u\in(0,1)$, define the deterministic Gaussian quantile
\[
q_n^G(u):=\inf\{t\in\R:F_n^G(t)\ge u\}.
\]
The distribution of $M_n^G$ is continuous: for every $t$,
\[
\PP(M_n^G=t)
\le\sum_{h\in\cH_n}\PP(|G_{n,h}|=t)=0,
\]
because each $G_{n,h}$ is standard normal. Consequently,
$F_n^G\{q_n^G(u)\}=u$.

Since \eqref{eq:conditional-bootstrap-gaussian} holds, there is a deterministic sequence
$\zeta_n\downarrow0$ such that
\[
\PP(\kappa_n>\zeta_n)\longrightarrow0.
\]
For the fixed $\alpha\in(0,1)$, take $n$ sufficiently large that
$\zeta_n<\min(\alpha,1-\alpha)$. Define the high-probability event
\[
\mathcal K_n
:=\mathcal E_{\nu,n}
\cap\left\{\min_{j\le p}\hat\tau_j^2\ge c_\tau/2\right\}
\cap\{\kappa_n\le\zeta_n\}.
\]
Equation~\eqref{eq:uniform-tau-positive},
\eqref{eq:studentization-positivity}, and
\eqref{eq:conditional-bootstrap-gaussian} imply $\PP(\mathcal K_n)\to1$.
On $\mathcal K_n$, uniform closeness of $F_n^\flat$ and $F_n^G$ gives
\begin{equation}
q_n^G(1-\alpha-\zeta_n)
\le c_{1-\alpha}^\flat
\le q_n^G(1-\alpha+\zeta_n).
\label{eq:bootstrap-quantile-bracket}
\end{equation}
For the upper bound, at $q_n^G(1-\alpha+\zeta_n)$ we have
$F_n^\flat\ge F_n^G-\zeta_n=1-\alpha$, so the definition of the bootstrap quantile gives
the claimed inequality. For the lower bound, if
$t<q_n^G(1-\alpha-\zeta_n)$, then
$F_n^\flat(t)\le F_n^G(t)+\zeta_n<1-\alpha$, so the bootstrap quantile cannot be smaller
than $q_n^G(1-\alpha-\zeta_n)$.

On $\mathcal K_n$, all $\hat\sigma^{\mathrm c}_{jS}$ are positive. Therefore the simultaneous
coverage event in \eqref{eq:simultaneous-coverage}, denoted by $\mathcal C_n$, is exactly
\[
\mathcal C_n=\{M_n^T\le c_{1-\alpha}^\flat\}.
\]
Combining this identity with \eqref{eq:bootstrap-quantile-bracket} gives
\begin{align*}
\PP(\mathcal C_n)
&\ge
\PP\{M_n^T\le q_n^G(1-\alpha-\zeta_n)\}
-\PP(\mathcal K_n^c)\\
&\ge 1-\alpha-\zeta_n-r_{G,n}-\PP(\mathcal K_n^c),
\end{align*}
and
\begin{align*}
\PP(\mathcal C_n)
&\le
\PP\{M_n^T\le q_n^G(1-\alpha+\zeta_n)\}
+\PP(\mathcal K_n^c)\\
&\le 1-\alpha+\zeta_n+r_{G,n}+\PP(\mathcal K_n^c).
\end{align*}
Since $\zeta_n\to0$, $r_{G,n}\to0$, and $\PP(\mathcal K_n^c)\to0$, the two bounds prove
$\PP(\mathcal C_n)\to1-\alpha$, which is part~(ii).
\end{proof}

\paragraph{Interpretation of Theorem~\ref{thm:bootstrap-band}.}
Part (i) shows that the conditional multiplier distribution consistently replaces the
infeasible Gaussian approximation in Theorem~\ref{thm:uniform-gaussian}. The bootstrap retains
the dependence across the entire covariate-by-contrast grid through the estimated correlation
matrix $\hat\Sigma$. Part (ii) turns this feasible approximation into a simultaneous confidence
band for all $pM$ coefficients.

\paragraph{Restriction to deterministic subsets.}
The conclusions of Theorems~\ref{thm:uniform-gaussian} and \ref{thm:bootstrap-band} remain
valid when $\cH_n$ is replaced throughout by any deterministic nonempty subset
$\mathcal I_n\subseteq\cH_n$.  Indeed, if
$d_{\mathcal I,n}:=|\mathcal I_n|$ and
$\ell_{\mathcal I,n}:=\log(2nd_{\mathcal I,n})$, then
$d_{\mathcal I,n}\le d_n$ and $\ell_{\mathcal I,n}\le\ell_n$, while every uniform maximum
appearing in the assumptions and proofs can only decrease after restriction.  Hence no new
rate or approximation argument is required.  In particular, Theorem~\ref{thm:bootstrap-band}(ii)
applied to the restricted array gives, for every fixed $\alpha\in(0,1)$,
\begin{equation}
\PP\!\left\{
\max_{h\in\mathcal I_n}|T_h|
\le c_{1-\alpha}^\flat(\mathcal I_n)
\right\}\longrightarrow1-\alpha.
\label{eq:deterministic-subset-validity}
\end{equation}

\begin{remark}[Testing prespecified groups of modifier cells]
\label{rem:prespecified-groups}
Equation~\eqref{eq:deterministic-subset-validity} also yields a test of the joint null
\[
H_0(\mathcal I_n):
\theta^\ast_{jS}=0
\quad\text{for every }(j,S)\in\mathcal I_n,
\]
where $\mathcal I_n\subseteq\cH_n$ is a prespecified deterministic nonempty subset.
Reject $H_0(\mathcal I_n)$ when
\[
\max_{(j,S)\in\mathcal I_n}
\left|\frac{\sqrt n\,\hat b_{jS}}{\hat\sigma^{\mathrm c}_{jS}}\right|
>c^\flat_{1-\alpha}(\mathcal I_n).
\]
Under the assumptions of Theorem~\ref{thm:bootstrap-band}, the rejection probability
converges to $\alpha$ under this joint null, because $T^0_{jS}=T_{jS}$ throughout
$\mathcal I_n$.

Two choices are particularly useful. For a prespecified contrast $S_0$, taking
$\mathcal I_n=[p]\times\{S_0\}$ tests whether its projection has any nonzero modifier
coefficient; when $|S_0|\ge2$, this concerns heterogeneity of an interaction effect.
For a prespecified covariate $j_0$, taking $\mathcal I_n=\{j_0\}\times\cS$ tests whether
that covariate has a nonzero partial projection coefficient for any factorial contrast.
Both hypotheses concern the chosen feature dictionary and do not exclude nonlinear
heterogeneity outside its linear span.

The restricted critical value is computed from the same multiplier draws by taking the
maximum only over $\mathcal I_n$, and satisfies
\[
c^\flat_{1-\alpha}(\mathcal I_n)
\le c^\flat_{1-\alpha}(\cH_n).
\]
The projection dictionary and fitted coefficients remain unchanged. This validity statement
applies to a prespecified group; it does not justify selecting a group after inspecting the
results. Separately testing several groups at level $\alpha$ does not by itself provide
familywise error control across those groups.
\end{remark}

\begin{corollary}[Strong familywise error control by Romano--Wolf step-down]
\label{cor:strong-fwer}
Under the assumptions of Theorem~\ref{thm:bootstrap-band}, fix $\alpha\in(0,1)$.  Using the
multiplier vector $W^\flat$ and its conditional law $\PP^\flat$ defined above, for every
nonempty $H\subseteq\cH_n$ define
\[
c_{1-\alpha}^\flat(H)
:=\inf\left\{t\in\R:
\PP^\flat\!\left(\max_{h\in H}|W_h^\flat|\le t\right)\ge1-\alpha
\right\}.
\]
The Romano--Wolf step-down algorithm starts from $H^{(1)}=\cH_n$.  Given a nonempty active
set $H^{(r)}$, choose
\[
h_r\in\operatorname*{arg\,max}_{h\in H^{(r)}}|T_h^0|,
\]
using a fixed deterministic rule to break ties.  If
$|T_{h_r}^0|\le c_{1-\alpha}^\flat(H^{(r)})$, stop and retain all hypotheses in $H^{(r)}$.
Otherwise reject $h_r$, set
$H^{(r+1)}=H^{(r)}\setminus\{h_r\}$, and continue, stopping after the current iteration if no
hypothesis remains.  For the deterministic true-null set
$\cH_{0,n}:=\{(j,S)\in\cH_n:\theta_{jS}^\ast=0\}$, this algorithm controls the strong
familywise error rate:
\[
\limsup_{n\to\infty}\PP\{\text{at least one index in $\cH_{0,n}$ is rejected}\}\le\alpha.
\]
\end{corollary}

\begin{proof}
If $\cH_{0,n}$ is empty, the familywise error probability is zero.  Suppose henceforth that
$\cH_{0,n}$ is nonempty.

\emph{Step 1: reduction to the maximum over the true nulls.}
For the comparison of bootstrap maxima, fix the realized $\mathcal D_n$.  The observed
statistics $T_h^0$ and the recursively constructed active sets $H^{(r)}$ are then fixed, while
$W^\flat$ varies according to $\PP^\flat$.
Consider a realized run of the step-down algorithm in which at least one true null is rejected,
and let $r_\ast$ be the first iteration at which this occurs.  No true null has been removed
before iteration $r_\ast$, and hence
\[
\cH_{0,n}\subseteq H^{(r_\ast)}.
\]
For every multiplier draw and every two index sets $H_1\subseteq H_2$,
\[
\max_{h\in H_1}|W_h^\flat|
\le
\max_{h\in H_2}|W_h^\flat|.
\]
It follows pathwise that their conditional distribution functions satisfy
\[
\PP^\flat\!\left(\max_{h\in H_2}|W_h^\flat|\le t\right)
\le
\PP^\flat\!\left(\max_{h\in H_1}|W_h^\flat|\le t\right),
\qquad t\in\R,
\]
and therefore their conditional quantiles are monotone:
\begin{equation}
c_{1-\alpha}^\flat(H_1)
\le c_{1-\alpha}^\flat(H_2).
\label{eq:stepdown-critical-monotonicity}
\end{equation}
At iteration $r_\ast$, the selected index $h_{r_\ast}\in\cH_{0,n}$ satisfies
$|T_{h_{r_\ast}}^0|>c_{1-\alpha}^\flat(H^{(r_\ast)})$.  By
\eqref{eq:stepdown-critical-monotonicity}, this implies
\begin{equation}
\max_{h\in\cH_{0,n}}|T_h^0|
>c_{1-\alpha}^\flat(\cH_{0,n}).
\label{eq:first-false-rejection-inclusion}
\end{equation}
Thus it remains to control the probability of the event in
\eqref{eq:first-false-rejection-inclusion}.

\emph{Step 2: bootstrap validity on the deterministic true-null subset.}
The set $\cH_{0,n}$ is deterministic, and all assumptions of
Theorems~\ref{thm:uniform-gaussian} and \ref{thm:bootstrap-band} are inherited after
restriction to this set.  Thus \eqref{eq:deterministic-subset-validity}, applied with
$\mathcal I_n=\cH_{0,n}$, and the identity $T_h^0=T_h$ on $\cH_{0,n}$ give
\[
\PP\!\left\{
\max_{h\in\cH_{0,n}}|T_h^0|
>c_{1-\alpha}^\flat(\cH_{0,n})
\right\}\longrightarrow\alpha.
\]
Combining this conclusion with \eqref{eq:first-false-rejection-inclusion} yields
\[
\PP\{\text{at least one index in $\cH_{0,n}$ is rejected}\}
\le\alpha+o(1).
\]
Taking the limit superior proves the stated strong familywise error control.  The event
reduction in Step~1 is the standard Romano--Wolf monotonicity argument
\citep{RomanoWolf2005}.
\end{proof}

\begin{remark}[Equivalent batched implementation]
\label{rem:batched-stepdown}
For a nonempty active set $H$, define
\[
R(H):=\{h\in H:|T_h^0|>c_{1-\alpha}^\flat(H)\}.
\]
A batched implementation rejects all indices in $R(H)$ at once and, when this set is nonempty,
continues from $H\setminus R(H)$.  It produces the same final rejection set as the sequential
procedure in Corollary~\ref{cor:strong-fwer}.  Indeed, if $R(H)$ is nonempty, every statistic
indexed by $R(H)$ is larger than every statistic indexed by $H\setminus R(H)$.  After any
subset of $R(H)$ has been removed, the remaining active set $H'\subseteq H$ satisfies
$c_{1-\alpha}^\flat(H')\le c_{1-\alpha}^\flat(H)$ by
\eqref{eq:stepdown-critical-monotonicity}.  Consequently, every index in $R(H)$ that remains
active still exceeds the updated critical value and is rejected before any index outside
$R(H)$ is considered.  The sequential and batched procedures therefore both reach
$H\setminus R(H)$; repeating this argument proves equality of their terminal rejection sets.
If $R(H)$ is empty, both procedures stop.  The same equivalence holds for the Monte Carlo
critical values $c_{1-\alpha}^{\flat,B}(H)$ when the same multiplier draws are reused across
steps, because the bootstrap maxima, and hence their empirical quantiles, retain the same
setwise monotonicity.
\end{remark}

\paragraph{How the conditions can be verified.}
The results above are deliberately stated in terms of the uniform remainder, moment,
and correlation bounds used in their proofs.
Section~\ref{sec:supp-uniform-sufficient} of the supplementary material gives
formal verification lemmas and a bounded-envelope primitive route. To summarize that route,
write $\ell_p=\log(2p)$,
$s_\theta=\max_S\|\theta_S^\ast\|_0$, and
$s_\gamma=\max_j\|\gamma_j\|_0$. Under uniformly bounded covariates, nodewise residuals,
and score residuals, well-conditioned covariance and nondegenerate variances, the usual
Lasso penalties of orders $\sqrt{\ell_n/n}$ and $\sqrt{\ell_p/n}$ reduce the high-level
conditions, among other things, to a transparent baseline requirement:
\[
\max_{\ell\le L}\{\mathbb P_{n,\ell}\Delta_{g,\ell}^2\}^{1/2}=o_p(\ell_n^{-2}),
\qquad
\max_{\ell\le L}\max_{i\in I_\ell}|\Delta_{g,\ell}(X_i)|=O_p(1),
\]
together with representative growth requirements
\[
\frac{\ell_n^7}{n}\to0,\qquad
s_\theta\ell_n\sqrt{\frac{\ell_p}{n}}\to0,\qquad
s_\gamma\ell_n\sqrt{\frac{\ell_p}{n}}\to0,
\]
and
\[
\ell_n^2\left\{
\sqrt{\frac{s_\theta\ell_n}{n}}+
\sqrt{\frac{s_\gamma\ell_p}{n}}
\right\}\to0.
\]
These conditions are sufficient rather than necessary. In particular, the bounded-envelope
route implicitly requires a bounded outcome residual and uniform $\ell_1$ control of the
projection and nodewise coefficients; when these restrictions are inappropriate, the
high-level assumptions of Section~\ref{sec:simul-conditions} should instead be verified
directly.

\paragraph{Fixed factorial design and growing covariate grid.}
Throughout the asymptotic analysis, $K$ and hence $M=2^K-1$ are fixed. Therefore
$d_n=pM$ grows only through $p=p_n$, and
$\ell_n=\log(2npM)=O(\log n+\log p)$. This formulation matches the intended medical
applications, in which the treatment components are chosen when the trial is designed rather
than added as the sample size grows. The number of contrasts still matters for finite-sample
precision and computation across different factorial designs: fitting all $M$ contrast
regressions has exponential cost in $K$. The simulations examine this finite-design effect by
comparing several fixed values of $K$; they do not represent growth of $K$ with $n$.

These results give simultaneous bands and strong FWER control, not FDR. FDR step-up rules
require relative tail accuracy at order-$\alpha/d_n$ thresholds and concentration of the
false discovery proportion, which the additive approximation here does not supply; the
companion paper develops the needed theory.

\subsection{Implementation of simultaneous inference}
\label{sec:simul-implementation}

We conclude this section by collecting the preceding estimators and bootstrap quantities into
an implementable procedure.  The inputs are the observed data, a fold partition, the Lasso and
nodewise-Lasso penalties, a nominal level $\alpha$, and a number $B$ of multiplier draws.  No
population quantity appearing in the preceding proofs is required.  The procedure is as
follows.

\begin{enumerate}[leftmargin=2.2em,label=\textbf{Step \arabic*.}]
\item \emph{Construct the cross-fitted factorial scores.}
Split the observations into $L$ folds.  For each fold $I_\ell$, estimate the treatment-invariant
baseline $\hat g^{(-\ell)}$ using only the observations outside that fold, and compute
$\hat\psi_{iS}$ from \eqref{eq:cfscore} for every $i\in I_\ell$ and $S\in\cS$.  The same fitted
baseline is used for all contrasts.

\item \emph{Fit the sparse projection and residualize the covariates.}
Center the score and covariates to obtain $\tilde\psi_{iS}$ and $\tilde Z_i$.  Fit one initial
Lasso regression for each contrast $S$ to obtain $\hat\theta_S$.  Separately, fit one nodewise
Lasso regression for each covariate $j$ and form $\hat V_{ij}$ and $\hat\tau_j^2$.  Because the
design matrix is common to all contrasts, these $p$ nodewise regressions are computed only once
and then reused for every $S$.

\item \emph{Debias and studentize every cell.}
For each $h=(j,S)\in\cH_n$, compute $\hat b_{jS}$ from \eqref{eq:debiased}.  Then form
\[
\hat\varepsilon_{iS}
:=\tilde\psi_{iS}-\tilde Z_i^\top\hat\theta_S,
\qquad
\hat u_{i,h}
:=\hat V_{ij}\hat\varepsilon_{iS}
-\mathbb P_n(\hat V_j\hat\varepsilon_S),
\qquad
\hat\nu_h^2:=\mathbb P_n\hat u_h^2.
\]
The standard error and observable null statistic are
\[
\hat\sigma^{\mathrm c}_{jS}=\frac{\hat\nu_h}{\hat\tau_j^2},
\qquad
T_{jS}^0=\frac{\sqrt n\,\hat b_{jS}}{\hat\sigma^{\mathrm c}_{jS}}.
\]
These quantities are evaluated when all $\hat\tau_j^2$ and $\hat\nu_h$ are positive; the
assumptions of Theorem~\ref{thm:bootstrap-band} imply that this event has probability tending
to one.  A zero or numerically negligible scale should be flagged in an implementation rather
than used to form a studentized statistic.

\item \emph{Generate the multiplier distribution.}
For $b=1,\ldots,B$, independently generate
\[
(e_1^{(b)},\ldots,e_n^{(b)})\mid\mathcal D_n\sim N(0,I_n),
\]
where $\mathcal D_n$ was defined above \eqref{eq:multiplier-bootstrap}, and compute
\[
W_h^{\flat(b)}
:=\frac1{\sqrt n}\sum_{i=1}^n
e_i^{(b)}\frac{\hat u_{i,h}}{\hat\nu_h},
\qquad h\in\cH_n.
\]
For every nonempty $H\subseteq\cH_n$, let
$c_{1-\alpha}^{\flat,B}(H)$ be the empirical $(1-\alpha)$-quantile of
\[
\left\{\max_{h\in H}|W_h^{\flat(b)}|:b=1,\ldots,B\right\}.
\]
This is the Monte Carlo counterpart of the conditional quantile
$c_{1-\alpha}^\flat(H)$ in \eqref{eq:bootstrap-critical-value}.

\item \emph{Construct the simultaneous confidence band.}
Using the full grid $H=\cH_n$, report
\begin{equation}
\left[
\hat b_{jS}\pm
c_{1-\alpha}^{\flat,B}(\cH_n)
\frac{\hat\sigma^{\mathrm c}_{jS}}{\sqrt n}
\right],
\qquad (j,S)\in\cH_n.
\label{eq:implemented-simultaneous-band}
\end{equation}
The common critical value accounts for dependence across all covariate-by-contrast cells.

\item \emph{Optionally perform Romano--Wolf step-down testing.}
To select cells while controlling the strong familywise error rate, initialize
$H^{(1)}=\cH_n$.  At iteration $r$, choose
\[
h_r\in\operatorname*{arg\,max}_{h\in H^{(r)}}|T_h^0|
\]
using the same deterministic tie-breaking rule as in Corollary~\ref{cor:strong-fwer}, and
compute $c_{1-\alpha}^{\flat,B}(H^{(r)})$.  If
$|T_{h_r}^0|\le c_{1-\alpha}^{\flat,B}(H^{(r)})$, stop.  Otherwise reject $h_r$, set
$H^{(r+1)}=H^{(r)}\setminus\{h_r\}$, and continue until the stopping rule is met or no index
remains.  The fitted models and multiplier coordinates are not recomputed: each iteration
reuses the same $B$ multiplier draws and changes only the set over which the maximum is taken.
Remark~\ref{rem:batched-stepdown} shows that the equivalent batched implementation gives the
same final rejection set.
\end{enumerate}

The simultaneous band in Step~5 provides joint uncertainty statements for all cells, whereas
Step~6 provides a multiplicity-adjusted selection rule.  The number $B$ controls only the
Monte Carlo accuracy with which the conditional bootstrap quantiles are evaluated; the
theoretical results above concern their ideal conditional counterparts.  Section~\ref{sec:sim}
uses $B=1000$ throughout.

\section{Simulation studies}\label{sec:sim}

The simulations evaluate the finite-sample accuracy and power of the simultaneous confidence
bands in Theorem~\ref{thm:bootstrap-band} and the Romano--Wolf step-down procedure in
Corollary~\ref{cor:strong-fwer}. The primary experiment varies the factorial dimension $K$,
the covariate dimension $p$, and the sample size $n$ separately, so that their distinct effects
are not hidden inside the total number of tested cells $d_n=p(2^K-1)$. Each value of $K$ is
treated as a separate fixed factorial design, so $K$ does not grow with $n$ within the
asymptotic framework. A second experiment holds the factorial grid fixed and varies the strength and quality
of the treatment-invariant baseline, thereby separating its role in validity from its role in
efficiency.

\subsection{Simulation design: factorial, covariate, and sample-size scaling}
\label{sec:sim-scaling}

\paragraph{Data-generating process.}
For each replication, independently generate
$A_i\sim\operatorname{Unif}(\{-1,+1\}^K)$ and
$Z_i\sim N_p\{0,\Sigma_Z(\rho)\}$, where $\Sigma_Z(\rho)$ is the AR(1) covariance
matrix with entries $\{\Sigma_Z(\rho)\}_{jk}=\rho^{|j-k|}$.
We set $\rho=0.5$ throughout the simulations. The outcome is
\begin{equation}\label{eq:sim-dgp}
Y_i=m_0(Z_i)+\frac12\sum_{S\in\cS}\tau_S(Z_i)\phi_S(A_i)+\epsilon_i,
\qquad
\epsilon_i\stackrel{\mathrm{iid}}{\sim}N(0,1),
\end{equation}
independently of $(A_i,Z_i)$, with
\[
m_0(z)=1+z_1+\frac12(z_2^2-1)+\frac14z_3z_4,
\qquad
\tau_S(z)=\eta_S+\sum_{j=1}^p\theta_{jS}^\ast z_j.
\]
We set $\eta_S=0.5$ for
$S\in\cC_0:=\{\{1\},\{2\},\{1,2\}\}$ and $\eta_S=0$ otherwise. Hence the global modifier-null
configuration below still allows nonzero average factorial effects. Because the modification
functions are linear and $Z$ is centered, the generating coefficient $\theta_{jS}^\ast$ is
exactly the projection target defined in Section~\ref{sec:setup}.

\paragraph{Coefficient configurations.}
For each design point, we consider two configurations of the population modifier coefficients.
In the \emph{global modifier-null configuration}, $\theta_{jS}^\ast=0$ for every
$(j,S)\in\cH_n$. In the \emph{sparse signal configuration}, the set of cells assigned
nonzero coefficients is
\[
\cH_{\mathrm{sig}}:=\bigl\{
(1,\{1\}),(2,\{1\}),(3,\{2\}),(4,\{3\}),
(5,\{1,2\}),(6,\{1,3\}),(7,\{2,3\}),(8,\{1,2,3\})
\bigr\}.
\]
This set contains main-effect, two-way, and three-way modifier cells and is available for every
$K\ge3$ and $p\ge20$. Coefficients outside $\cH_{\mathrm{sig}}$ are zero, so this configuration contains
both true and false cellwise null hypotheses. Assign alternating signs $s_h\in\{-1,+1\}$
over $h\in\cH_{\mathrm{sig}}$ and set
\begin{equation}\label{eq:sim-signal}
\theta_h^\ast
=s_h\,\delta\,\sigma_h^{(0)}
\sqrt{\frac{2\log(2d_n)}{n}},
\qquad h\in\cH_{\mathrm{sig}}.
\end{equation}
Here $\sigma_h^{(0)}$ is the population
asymptotic standard deviation in the corresponding global modifier-null configuration. It is available in
closed form for the present DGP. Write
\[
\Omega(\rho):=\Sigma_Z(\rho)^{-1},
\qquad
\omega_j(\rho):=\{\Omega(\rho)\}_{jj}.
\]
Under the global modifier-null configuration and the oracle baseline $g=m_0$, the projection residual is
independent of $Z$. Walsh--Hadamard orthogonality then gives
\begin{equation}\label{eq:sim-null-sd}
\left(\sigma_{jS}^{(0)}\right)^2
=\omega_j(\rho)\left(4+\sum_{\substack{U\in\cS\\U\ne S}}\eta_U^2\right),
\qquad
\omega_j(\rho)=
\begin{cases}
(1-\rho^2)^{-1}, & j\in\{1,p\},\\
(1+\rho^2)(1-\rho^2)^{-1}, & 2\le j\le p-1.
\end{cases}
\end{equation}
These population standard deviations are calculated analytically and fixed before the Monte
Carlo experiment; they are never estimated from the simulated samples. The reference signal
level is $\delta=1$; at the reference design $(K,p,n)=(4,50,1000)$ we also use
$\delta\in\{0.75,1.25\}$ to display a nondegenerate power curve. Thus signal strength is tied
to the simultaneous detection scale rather than fixed while $d_n$ changes.

\paragraph{Dimension and sample-size configurations.}
The reference design is $(K,p,n)=(4,50,1000)$. We vary one dimension at a time around this
reference:
\begin{table}[H]
\centering
\small
\begin{tabular}{lcccc}
\hline
Panel & $K$ & $p$ & $n$ & $d_n=p(2^K-1)$\\
\hline
Factorial dimension & $3,4,5,6$ & $50$ & $1000$ & $350,750,1550,3150$\\
Covariate dimension & $4$ & $20,50,100,200$ & $1000$ & $300,750,1500,3000$\\
Sample size & $4$ & $50$ & $200,500,1000$ & $750$\\
\hline
\end{tabular}
\caption{Dimension and sample-size configurations for the primary simulation experiment. The
reference design appears in all three one-at-a-time panels but is simulated only once.}
\label{tab:sim-configurations}
\end{table}

\paragraph{Implementation.}
We use the oracle treatment-invariant baseline
$\hat g^{(-\ell)}=m_0$ on each of $L=5$ folds. This isolates the Gaussian and
multiplier-bootstrap approximations from baseline-estimation error; the effect of estimating
or misspecifying the baseline is examined separately. The score Lasso uses
\[
\lambda_S=1.1\,\hat s_S\sqrt{\frac{2\log(pM)}{n}},
\]
where $\hat s_S$ is the empirical standard deviation of the centered score for contrast $S$,
and every nodewise regression uses
$\lambda_j=1.1\sqrt{2\log(p)/n}$. The multiplier bootstrap uses $B=1000$ Gaussian draws and
all procedures are calibrated at $\alpha=0.05$. We use $1000$ Monte Carlo replications per
configuration and report Monte Carlo standard errors.

\paragraph{Performance measures.}
Under the global modifier-null configuration, the Romano--Wolf procedure makes at least one rejection if and
only if the multiplier-bootstrap band fails to cover zero in at least one cell. Thus its FWER
equals one minus the familywise coverage probability, and we report only the FWER to avoid
duplicating the same calibration measure. The analogous identity holds for Holm testing and
the Bonferroni band. Under the sparse signal configuration, the two criteria are no longer equivalent:
we report (i) familywise coverage of the simultaneous band, (ii) the strong FWER of the
step-down procedure over the true nulls, and (iii) the average true-positive rate over $\cH_{\mathrm{sig}}$.
For both coefficient configurations, we also report the average band half-width and the mean
bootstrap critical value $c_{0.95}^\flat$. Bonferroni simultaneous bands and Holm step-down
testing, computed from the same debiased estimates and standard errors, serve as
dependence-agnostic benchmarks. To display the comparison directly, we summarize the gain of
the proposed procedures by
\[
1-\frac{\text{average multiplier-bootstrap half-width}}
        {\text{average Bonferroni half-width}}
\quad\text{and}\quad
\operatorname{TPR}_{\mathrm{RW}}-\operatorname{TPR}_{\mathrm{Holm}},
\]
where the second quantity is reported under the sparse signal configuration. Because the underlying
estimates and standard errors are held fixed across procedures, these differences isolate the
effect of the simultaneous calibration method.

\paragraph{Finite-sample calibration and dimensional scaling.}
Table~\ref{tab:sim-primary-results} reports the primary results at the reference signal level
$\delta=1$.  The Romano--Wolf FWER under the global modifier-null configuration ranges from $0.045$ to
$0.059$ across the nine distinct designs, and is therefore close to the nominal $0.05$ level.
Under the sparse signal configuration, simultaneous coverage ranges from $0.940$ to $0.953$ and
the strong FWER ranges from $0.046$ to $0.057$.  The Monte Carlo standard errors of these
probabilities are at most $0.008$, so the small deviations from the nominal levels are
consistent with ordinary simulation variation and mild finite-sample approximation error.

\begin{table}[H]
\centering
\small
\setlength{\tabcolsep}{3.5pt}
\begin{tabular}{@{}llrrrrrrr@{}}
\hline
& & & Global null & \multicolumn{5}{c}{Sparse signal configuration}\\
\cline{4-4}\cline{5-9}
Panel & Varied value & $d_n$ & FWER & Coverage & FWER & TPR & Half-width & $c_{0.95}^\flat$\\
\hline
Factorial & $K=3$ & 350  & 0.049 & 0.943 & 0.056 & 0.427 & 0.320 & 3.792\\
dimension & $K=4$ & 750  & 0.047 & 0.953 & 0.047 & 0.425 & 0.340 & 3.977\\
          & $K=5$ & 1550 & 0.047 & 0.949 & 0.051 & 0.417 & 0.358 & 4.145\\
          & $K=6$ & 3150 & 0.054 & 0.951 & 0.049 & 0.414 & 0.375 & 4.306\\
\hline
Covariate & $p=20$  & 300  & 0.051 & 0.951 & 0.051 & 0.426 & 0.322 & 3.754\\
dimension & $p=50$  & 750  & 0.047 & 0.953 & 0.047 & 0.425 & 0.340 & 3.977\\
          & $p=100$ & 1500 & 0.045 & 0.953 & 0.046 & 0.419 & 0.354 & 4.137\\
          & $p=200$ & 3000 & 0.053 & 0.947 & 0.053 & 0.434 & 0.367 & 4.295\\
\hline
Sample size & $n=200$  & 750 & 0.059 & 0.940 & 0.057 & 0.205 & 0.878 & 3.958\\
            & $n=500$  & 750 & 0.048 & 0.944 & 0.057 & 0.359 & 0.498 & 3.975\\
            & $n=1000$ & 750 & 0.047 & 0.953 & 0.047 & 0.425 & 0.340 & 3.977\\
\hline
\end{tabular}
\caption{Primary simulation results at $\delta=1$ based on $1000$ Monte Carlo replications.
The ``Global null'' column reports the Romano--Wolf familywise error rate under the global
modifier-null configuration. The five columns under ``Sparse signal configuration'' report
familywise coverage of the
multiplier-bootstrap band, strong Romano--Wolf FWER, true-positive rate over the eight signal
cells, average band half-width, and average bootstrap critical value.  The reference design is
repeated across panels for readability but is simulated only once.}
\label{tab:sim-primary-results}
\end{table}

Increasing either $K$ or $p$ raises the number of simultaneous targets and hence the critical
value and band width.  Along the factorial-dimension panel, $c_{0.95}^\flat$ increases from
$3.792$ to $4.306$ and the average half-width increases from $0.320$ to $0.375$, while the TPR
decreases only modestly from $0.427$ to $0.414$.  Along the covariate-dimension panel, coverage
and FWER remain stable through $p=200$; the nonmonotonic TPR between $0.419$ and $0.434$ is small
relative to its Monte Carlo standard error, which is approximately $0.005$.  These power
comparisons should be interpreted at the simultaneous detection scale: by
\eqref{eq:sim-signal}, the coefficient magnitude changes with $d_n$ rather than being held
fixed as the dimension grows.

The sample-size panel exhibits the clearest finite-sample improvement.  From $n=200$ to
$n=1000$, average half-width falls from $0.878$ to $0.340$, TPR rises from $0.205$ to $0.425$,
and the maximum absolute Monte Carlo bias over the signal cells falls from $0.192$ to $0.048$.
This improvement occurs even though the nonzero coefficients shrink at the rate in
\eqref{eq:sim-signal}; it therefore reflects improved estimation and Gaussian calibration, not
an increase in the unstandardized signal.

At the reference design, increasing $\delta$ from $0.75$ to $1$ and $1.25$ raises the TPR from
$0.149$ to $0.425$ and $0.711$, respectively.  The corresponding band coverages are $0.943$,
$0.953$, and $0.968$, and the strong FWERs are $0.055$, $0.047$, and $0.032$.  The last design
is somewhat conservative, but all three settings maintain familywise error control while
displaying a nondegenerate power curve.

\paragraph{Comparison with dependence-agnostic calibration.}
The multiplier and Bonferroni critical values are close in this DGP.  Across all primary
configurations, the multiplier band is between $0.24\%$ and $0.75\%$ shorter than the
Bonferroni band, with an average reduction of $0.34\%$.  Under the sparse signal configuration, the
Romano--Wolf TPR exceeds the Holm TPR by $0.3$ to $1.0$ percentage points, averaging $0.5$
percentage points.  Thus the numerical benefit of estimating dependence is modest in this
particular Gaussian AR(1) design.  The main evidence from this experiment is the calibration of
the feasible high-dimensional procedure as $d_n$ increases, rather than a large efficiency
gain over Bonferroni or Holm.

\subsection{Sensitivity to baseline estimation and misspecification}
\label{sec:sim-baseline}

Section~\ref{sec:sim-scaling} uses the oracle baseline to isolate the high-dimensional Gaussian
and bootstrap approximations. We now separate two claims about baseline residualization. The
projection target is exactly baseline-free for every treatment-invariant baseline by
Theorem~\ref{thm:ident}. Valid calibration with an estimated baseline additionally requires
cross-fitting and the weighted-consistency and uniform-remainder conditions imposed in
Sections~\ref{sec:debias} and~\ref{sec:simul}. Within this admissible class, the choice of
baseline can still affect interval length and power through its predictive accuracy.

\paragraph{Design and prognostic strength.}
We retain the data-generating process in \eqref{eq:sim-dgp}, the reference dimensions
$(K,p)=(4,50)$, and the global modifier-null and sparse signal configurations defined above. To
vary the amount
of prognostic variation, replace $m_0$ in \eqref{eq:sim-dgp} by $\xi m_0$, where
$\xi\in\{0,1,2\}$ at $n=1000$. We additionally use $n\in\{200,500,1000\}$ at
$\xi=1$. The duplicated reference configuration $(n,\xi)=(1000,1)$ is simulated only once.
Under balanced assignment,
\[
\E(Y_i\mid Z_i=z)=g_{\mathrm{opt},\xi}(z)=\xi m_0(z),
\]
so $\xi$ changes the scope for efficiency gains from residualization, while all factorial
effects remain fixed. Under the sparse signal configuration, we use the signal rule in
\eqref{eq:sim-signal} with $\delta=1$ and the closed-form standard deviation in
\eqref{eq:sim-null-sd}. This standard deviation does not depend on $\xi$, because
residualization by $g_{\mathrm{opt},\xi}=\xi m_0$ removes the entire prognostic component.
Consequently, the generating coefficient vector is identical across all baseline procedures
within a design point.

\paragraph{Baseline procedures.}
We compare the following treatment-invariant baselines; every estimated baseline regresses
$Y$ on functions of $Z$ while ignoring $A$.
\begin{enumerate}[label=(\roman*),leftmargin=2.2em]
\item \emph{Oracle}: $\hat g^{(-\ell)}(z)=\xi m_0(z)$.
\item \emph{No residualization}: $\hat g^{(-\ell)}(z)=0$.
\item \emph{Cross-fitted linear}: ordinary least squares of $Y$ on $(1,Z^\top)^\top$ using
      $I_{-\ell}$. This baseline deliberately omits the quadratic and interaction terms in
      $m_0$.
\item \emph{Cross-fitted quadratic Lasso}: a Lasso regression \citep{Tibshirani1996} on the
      dictionary
      \[
      D(z)=\bigl(z_1,\ldots,z_p,z_1^2-1,\ldots,z_p^2-1,
      (z_jz_k)_{1\le j<k\le p}\bigr)^\top,
      \]
      with an unpenalized intercept. Its penalty is selected by five-fold cross-validation
      within $I_{-\ell}$, so the validation fold used to form the score is never used to fit or
      tune the baseline. The dictionary contains the true $m_0$ but is high-dimensional.
\end{enumerate}
The oracle and no-residualization choices are fixed and introduce no baseline-estimation
remainder. The cross-fitted linear fit converges to a deterministic linear projection even
though that projection is misspecified relative to $g_{\mathrm{opt},\xi}$; such
misspecification is not itself a violation because the theory does not require convergence to
the variance-optimal baseline. The quadratic dictionary, by contrast, contains the true
$\xi m_0$ and provides a flexible cross-fitted estimate. These four procedures therefore
represent fixed or admissible cross-fitted baselines under the present DGP.

At the small-sample configuration $(n,\xi)=(200,1)$, we add a diagnostic in which the same
quadratic Lasso is fitted and tuned on the full sample. This no-cross-fitting version lies
outside the conditions of Lemma~\ref{lem:global} and is included only to assess the practical
importance of sample splitting; it is not treated as a competing valid procedure.

\paragraph{Implementation and performance measures.}
Within each replication, the assignment, covariates, errors, outer folds, and multiplier draws
are shared across baseline procedures. All score-Lasso, nodewise-regression, bootstrap, and
Monte Carlo settings are otherwise identical to Section~\ref{sec:sim-scaling}. Under the global
modifier-null configuration, we report the Romano--Wolf FWER; the complementary familywise
coverage is not
tabulated separately. Under the sparse signal configuration, we report familywise band coverage, strong
FWER over the true nulls, the average true-positive rate over $\cH_{\mathrm{sig}}$, and the maximum absolute
Monte Carlo bias over its eight signal cells. For both configurations, efficiency is summarized
by the average simultaneous-band half-width and its ratio to the oracle half-width. For the
estimated baselines we also report the out-of-fold prediction error
\[
\frac1n\sum_{\ell=1}^L\sum_{i\in I_\ell}
\{\hat g^{(-\ell)}(Z_i)-\xi m_0(Z_i)\}^2.
\]
Within the fixed and admissible cross-fitted procedures, this design distinguishes failure of
simultaneous calibration from a purely efficiency-related loss due to a baseline that is
misspecified relative to $g_{\mathrm{opt},\xi}$. The no-cross-fitting diagnostic separately
examines what happens when the sample-splitting requirement itself is violated.

\paragraph{Results: prognostic strength and baseline quality.}
For the fixed and admissible cross-fitted baselines, familywise coverage under the sparse signal
configuration ranges from $0.933$ to $0.959$, and the Romano--Wolf FWER under either coefficient
configuration ranges from $0.035$ to $0.066$.  The Monte Carlo standard errors are at most
$0.008$ for these probabilities.  Thus the baseline choice has little effect on calibration,
as predicted by the baseline-free target and the cross-fitted theory, but it has a substantial
effect on precision and power.

Table~\ref{tab:sim-prognostic-results} displays this efficiency effect at $n=1000$.  When
$\xi=0$, no residualization coincides with the oracle and the cross-fitted quadratic Lasso is
also essentially oracle-efficient.  As the nonlinear prognostic component becomes stronger,
the quadratic Lasso remains close to the oracle: at $\xi=1$ and $2$, its bands are only about
$1.8\%$ and $1.9\%$ wider, and its TPRs are $0.396$ and $0.392$, compared with oracle TPRs of
$0.429$ and $0.422$.  Its prediction error is approximately $0.048$ in both settings.

\begin{table}[H]
\centering
\small
\setlength{\tabcolsep}{4pt}
\begin{tabular}{@{}crrrrrrrr@{}}
\hline
& \multicolumn{4}{c}{True-positive rate} & \multicolumn{4}{c}{Half-width relative to oracle}\\
\cline{2-5}\cline{6-9}
$\xi$ & Oracle & Quadratic CF & Linear CF & None & Oracle & Quadratic CF & Linear CF & None\\
\hline
0 & 0.425 & 0.422 & 0.371 & 0.425 & 1.000 & 1.001 & 1.035 & 1.000\\
1 & 0.429 & 0.396 & 0.152 & 0.031 & 1.000 & 1.018 & 1.264 & 1.788\\
2 & 0.422 & 0.392 & 0.032 & 0.003 & 1.000 & 1.019 & 1.778 & 3.126\\
\hline
\end{tabular}
\caption{Effect of prognostic strength at $(K,p,n)=(4,50,1000)$ under the sparse signal
configuration.  ``CF'' denotes cross-fitting.  Half-width ratios are computed replication by
replication relative to the oracle procedure and then averaged.}
\label{tab:sim-prognostic-results}
\end{table}

The misspecified linear baseline is valid but increasingly inefficient.  At $\xi=1$, its
average half-width is $1.264$ times the oracle width and its TPR is $0.152$; at $\xi=2$, these
quantities become $1.778$ and $0.032$.  No residualization loses still more precision, reaching
a width ratio of $3.126$ and a TPR of $0.003$ at $\xi=2$.  This contrast supports the theoretical
separation between validity and efficiency: treatment-invariant misspecification need not
invalidate inference, but accurate prediction of the prognostic mean can be essential for
detecting effect modification.

\paragraph{Results: sample size and cross-fitting.}
Table~\ref{tab:sim-baseline-sample} compares the cross-fitted quadratic Lasso with the oracle at
$\xi=1$.  Its efficiency gap narrows with sample size: the oracle half-width ratio falls from
$1.095$ at $n=200$ to $1.039$ at $n=500$ and $1.018$ at $n=1000$, while prediction error falls
from $0.408$ to $0.118$ and $0.048$.  Its TPR correspondingly rises from $0.131$ to $0.305$ and
$0.396$.  The maximum bias is nearly identical to the oracle at each sample size, indicating
that the remaining small-sample bias is not generated by cross-fitted baseline estimation.

\begin{table}[H]
\centering
\small
\setlength{\tabcolsep}{4pt}
\begin{tabular}{@{}rlrrrrrr@{}}
\hline
$n$ & Baseline & Coverage & FWER & TPR & Width/oracle & Prediction error & Max. bias\\
\hline
200  & Oracle          & 0.933 & 0.061 & 0.204 & 1.000 & --    & 0.197\\
     & Quadratic CF    & 0.942 & 0.057 & 0.131 & 1.095 & 0.408 & 0.197\\
     & Quadratic no CF & 0.902 & 0.064 & 0.147 & 0.952 & 0.287 & 0.256\\
\hline
500  & Oracle          & 0.943 & 0.056 & 0.357 & 1.000 & --    & 0.082\\
     & Quadratic CF    & 0.939 & 0.060 & 0.305 & 1.039 & 0.118 & 0.082\\
\hline
1000 & Oracle          & 0.947 & 0.050 & 0.429 & 1.000 & --    & 0.039\\
     & Quadratic CF    & 0.955 & 0.043 & 0.396 & 1.018 & 0.048 & 0.039\\
\hline
\end{tabular}
\caption{Baseline comparison over sample size at $(K,p,\xi)=(4,50,1)$ under the sparse signal
configuration.  Coverage is for the simultaneous band, FWER is strong Romano--Wolf FWER, and
maximum bias is taken over the eight signal cells.  The no-cross-fitting procedure at $n=200$
is a diagnostic outside the assumptions of the theory.}
\label{tab:sim-baseline-sample}
\end{table}

The no-cross-fitting diagnostic explains why sample splitting is not merely a technical device.
At $n=200$, fitting and tuning the quadratic Lasso on the score observations themselves produces
an apparently favorable width ratio of $0.952$, but simultaneous coverage falls to $0.902$ and
maximum absolute bias rises to $0.256$.  The coverage shortfall is more than five Monte Carlo
standard errors below $0.95$.  By comparison, cross-fitting raises coverage to $0.942$ and
reduces maximum bias to $0.197$, at the cost of the expected small-sample efficiency loss.  The
cross-fitted quadratic procedure still has a mild finite-sample distortion under the global
modifier-null configuration at $n=200$, where its FWER is $0.066$ (about two Monte Carlo
standard errors above $0.05$),
but this distortion is absent at the larger sample sizes.  Overall, the experiment supports the
claimed validity--efficiency separation and shows that the practical quadratic baseline
approaches oracle efficiency as $n$ increases.

\section{Discussion and extensions}\label{sec:disc}

This paper develops a unified route from heterogeneous factorial effects to simultaneous
inference over the covariate-by-contrast grid.  The central device is the residualized
Walsh--Hadamard score.  Because the assignment probabilities are known, subtracting any
treatment-invariant baseline leaves the conditional mean of every nonempty factorial contrast
unchanged.  Cross-fitting then turns baseline estimation error into an exactly conditionally
centered perturbation within each validation fold.  This structure permits all observations to
contribute to every contrast and separates the role of the baseline in identification from its
role in precision.

This separation should not be interpreted as saying that an arbitrary estimated baseline is
asymptotically harmless.  The pointwise theory still requires foldwise weighted consistency,
and the score Lasso requires the baseline-induced empirical process to remain below its penalty
scale.  The simultaneous results additionally require uniform linearization, studentization,
and high-dimensional moment conditions.  What the randomization identity removes is a
prescribed polynomial baseline rate and the usual product-rate restriction between two
estimated nuisance functions.  Thus the practical message is to cross-fit a flexible
treatment-invariant predictor and to regard its predictive quality mainly as an efficiency
choice, while still checking the regularity conditions needed by the score regression.

The simulations support this interpretation.  Simultaneous coverage and Romano--Wolf strong
FWER remain close to their nominal levels across the primary designs, which range from $300$ to
$3150$ effect-modifier cells and vary the factorial dimension, covariate dimension, and sample
size.  Better baseline prediction translates into materially narrower bands and greater power:
the cross-fitted quadratic Lasso approaches the oracle as $n$ grows and is within about $2\%$ of
the oracle width at $n=1000$ when prognostic variation is present.  By contrast, omitting
cross-fitting at $n=200$ produces shorter-looking intervals but simultaneous coverage of only
$0.902$.  This diagnostic underscores that sample splitting is part of the inferential design,
not merely a technical convenience.  The modest improvement over Bonferroni--Holm in the
Gaussian designs also gives an appropriately limited empirical conclusion: the multiplier
bootstrap adapts to dependence and is well calibrated here, but need not yield a large power
gain in every covariance structure.

The procedure is computationally feasible despite the size of the inferential grid.  All
contrasts share the same covariate design, so the nodewise regressions are computed once rather
than once per contrast.  The same multiplier draws can also be reused to construct the
simultaneous band and throughout the Romano--Wolf step-down algorithm.  Accordingly, the main
incremental cost of adding factorial contrasts lies in forming and fitting the contrast-specific
scores, not in repeating the design-side orthogonalization.

The inferential target should follow the scientific question. An investigator may prespecify
a particular treatment interaction and ask whether its effect varies with any of the included
baseline features, or prespecify a baseline characteristic and ask whether it modifies any
factorial contrast. The restricted tests in Remark~\ref{rem:prespecified-groups} address these
questions using a critical value calibrated to the relevant cells. When the objective is to
report effect-modification patterns across the entire grid, the simultaneous band provides
joint uncertainty statements; when the objective is to identify individual modifier cells with
strong familywise error control, the Romano--Wolf step-down procedure provides a corresponding
selection rule. A rejection of a group null establishes evidence of at least one nonzero
projection coefficient within that group. Attribution to particular cells should use
simultaneous intervals or multiplicity-adjusted cellwise tests.

Several limitations qualify these conclusions.  First, the estimand $\theta_{jS}^\ast$ is a
population projection coefficient defined by the chosen
feature dictionary, its scaling, and the covariate distribution.  It summarizes linear effect
modification after adjustment for the other dictionary elements; it is not generally a
derivative, a conditional causal effect at a point, or an assumption that the true modification
surface is linear.  These choices should therefore be reported explicitly, and substantive
interpretation should distinguish a projection summary from a fully nonparametric modifier
effect.  The asymptotic theory also holds $K$ fixed, assumes known positive assignment
probabilities, and lets the grid grow through $p=p_n$.  It relies on sparse score and nodewise
regressions and, in the primitive verification, bounded envelopes.  Designs with growing $K$,
dense projection coefficients, or substantially heavier tails require different rate budgets
and are not covered by the present results.

Extensions beyond individually randomized experiments are also important.  Many factorial and
conjoint experiments contain repeated evaluations within respondent
\citep{HHY2014,HainmuellerHopkins2015}.  When respondents are independent sampling units,
averaging the profile-level residualized scores within respondent and using respondent-level
multipliers is a natural extension.  The randomization-based centering argument can survive
this aggregation, but the variance estimation and Gaussian approximation must be reformulated
at the cluster level.  Observational factorial treatments present a more fundamental change.
If $e_a(x)=\PP(A=a\mid X=x)$ is unknown, replacing it by $\hat e_a(x)$ destroys the exact
identity used in Lemma~\ref{lem:global}.  A natural route is a Neyman-orthogonal score whose
remainder depends on products of propensity and outcome-regression errors.  Establishing the
corresponding pointwise and simultaneous theory, possibly with factorized or low-order models
for the $2^K$ treatment probabilities, is left for future work.

A further direction is to exploit structure shared across modifier contrasts.  The current
estimator fits each contrast separately and uses sparsity within each column of
$\Theta^\ast=(\theta^\ast_{jS})_{j\in[p],S\in\cS}$.  If a small number of modifier patterns is
shared across many contrasts, joint penalties, group structure, or a low-rank factorization of
$\Theta^\ast$ may improve estimation and power.  Such regularization changes the debiasing
problem and requires new cellwise inference.  Likewise, the direct-score identity can be
combined with nonlinear learners for $\tau_S(x)$, but simultaneous inference would then have
to be formulated for a different, explicitly defined collection of nonlinear summaries.  A
held-out comparison between the fitted linear projection and a flexible score regression can
serve as a useful diagnostic, although a formal calibration test is beyond the present paper.

Robustness to design imbalance and tail behavior is another important issue.  Positive but
unequal assignment probabilities are allowed for fixed $K$, yet severe imbalance
inflates the inverse-probability multipliers and can make the bands too wide to be informative.
Balanced allocation is therefore the natural primary design, with strongly unbalanced designs
treated as sensitivity analyses.  Extending the primitive theory from bounded envelopes to
sub-Gaussian or heavier-tailed scores would require correspondingly stronger control of the
Gaussian approximation and variance estimators, potentially together with truncation or robust
score construction.

Finally, the choice of multiplicity criterion marks a substantive boundary of the present
analysis.  This paper establishes simultaneous coverage and strong FWER control, not false
discovery rate control.  FDR procedures probe farther into the tail of the studentized maximum
and require moderate-deviation and dependence conditions beyond the additive Gaussian
approximation proved here; those questions are addressed separately in the companion paper.

Within this scope, the results show that factorial randomization supplies more than
identification of average contrasts.  Combined with a shared residualized score, cross-fitted
prediction, and high-dimensional bootstrap calibration, it also provides a practical basis for
familywise-valid inference on which baseline covariates modify which factorial effects.

\section*{Acknowledgement}
The authors used GPT-6 Astra (OpenAI) to translate the manuscript from Japanese into English. The authors subsequently reviewed and edited the translated text and take full responsibility for the content of this publication.

\clearpage
\appendix
\renewcommand{\thesection}{S\arabic{section}}
\renewcommand{\thesubsection}{\thesection.\arabic{subsection}}
\renewcommand{\theequation}{S\arabic{equation}}
\renewcommand{\thetheorem}{S\arabic{theorem}}
\renewcommand{\theproposition}{S\arabic{proposition}}
\renewcommand{\thelemma}{S\arabic{lemma}}
\renewcommand{\thecorollary}{S\arabic{corollary}}
\renewcommand{\theassumption}{S\arabic{assumption}}
\renewcommand{\theremark}{S\arabic{remark}}
\renewcommand{\thedefinition}{S\arabic{definition}}
\providecommand{\theHsection}{}
\providecommand{\theHequation}{}
\providecommand{\theHtheorem}{}
\providecommand{\theHproposition}{}
\providecommand{\theHlemma}{}
\providecommand{\theHcorollary}{}
\providecommand{\theHassumption}{}
\providecommand{\theHremark}{}
\providecommand{\theHdefinition}{}
\renewcommand{\theHsection}{supp.section.\arabic{section}}
\renewcommand{\theHequation}{supp.equation.\arabic{equation}}
\renewcommand{\theHtheorem}{supp.theorem.\arabic{theorem}}
\renewcommand{\theHproposition}{supp.proposition.\arabic{proposition}}
\renewcommand{\theHlemma}{supp.lemma.\arabic{lemma}}
\renewcommand{\theHcorollary}{supp.corollary.\arabic{corollary}}
\renewcommand{\theHassumption}{supp.assumption.\arabic{assumption}}
\renewcommand{\theHremark}{supp.remark.\arabic{remark}}
\renewcommand{\theHdefinition}{supp.definition.\arabic{definition}}
\setcounter{section}{0}
\setcounter{equation}{0}
\setcounter{theorem}{0}
\setcounter{proposition}{0}
\setcounter{lemma}{0}
\setcounter{corollary}{0}
\setcounter{assumption}{0}
\setcounter{remark}{0}
\setcounter{definition}{0}

\section*{Supplementary material}
This supplement provides primitive sufficient conditions for the pointwise and simultaneous
inference results in the main text.  Section~\ref{sec:primitive-pointwise} verifies the
pointwise nuisance rates, including the effect of the estimated baseline on the score Lasso.
Section~\ref{sec:supp-uniform-sufficient} gives empirical and population sufficient conditions
for uniform studentization, correlation estimation, and the simultaneous Gaussian and
multiplier-bootstrap approximations.


\section{Primitive sufficient conditions for the pointwise nuisance rates}
\label{sec:primitive-pointwise}

Assumption~\ref{as:reg}(P4) isolates the nuisance quantities used in the pointwise
linearization. We now derive those quantities from a population eigenvalue condition,
sparsity, tail conditions, and a foldwise conditional-Bernstein condition for the baseline
perturbation. For $r\in\{1,2\}$, define the Orlicz norm
\[
\|U\|_{\psi_r}:=\inf\{c>0:\E\exp(|U|^r/c^r)\le2\}.
\]
Let $T_S=\{k:\theta_{kS}^\ast\ne0\}$ and
$T_j=\{k\ne j:\gamma_{jk}\ne0\}$, with cardinalities $s_S=|T_S|$ and $s_j=|T_j|$.
Write
\[
\ell_p:=\log(2p),\qquad
\ell_d:=\log(2pM),\qquad
\ell_{np}:=\log(2np).
\]
Here $\ell_p$ is the logarithmic price of maximizing over the $p$ Lasso-score coordinates,
$\ell_d$ is the corresponding scale over the $pM$ coefficient cells, and $\ell_{np}$ is used
only to control the sample envelope over observations and coordinates. For a matrix
$B=(B_{km})$, define
$\|B\|_{\max}:=\max_{k,m}|B_{km}|$. The implicit constant in $\lesssim$ below may depend only
on the fixed constants displayed in the assumptions and on the limiting fold proportions, but
not on $n,p,$ or the target cell; it may depend on the fixed factorial design.
Because $K$ and the positive assignment probabilities are fixed, the design constant
$C_q:=\max_{S\in\cS}\max_{a\in\cA}|q_S(a)|$ is finite.

\begin{assumption}[Primitive pointwise sparse-regression conditions]
\label{as:primitive-pointwise}
For the deterministic target sequence $(j,S)$, the following conditions hold.
\begin{enumerate}[leftmargin=1.8em,label=(Q\arabic*)]
\item \emph{Centered sub-Gaussian design and population eigenvalues.}
For constants $K_Z<\infty$ and $0<c_\Sigma<C_\Sigma<\infty$,
\[
\E Z=0,
\qquad \E Z_k^2=1\quad(k\in[p]),
\qquad
\sup_{\|u\|_2=1}\|u^\top Z\|_{\psi_2}\le K_Z,
\]
\[
c_\Sigma\le\lambda_{\min}(\Sigma_Z)
\le\lambda_{\max}(\Sigma_Z)\le C_\Sigma.
\]

\item \emph{Sparse targets and oracle noise.}
$\|\theta_S^\ast\|_0=s_S$, $\|\gamma_j\|_0=s_j$, and
$\|\varepsilon_S\|_{\psi_2}\le K_\varepsilon$ for a constant $K_\varepsilon<\infty$.
The projection normal equations hold:
$\E(\varepsilon_S)=0$, $\E(Z\varepsilon_S)=0$, and
$\E(Z_{-j}V_j)=0$.

\item \emph{Conditional-Bernstein control of the baseline score.}
Define
\[
\Delta_{g,Z,n}^2
:=\max_{1\le\ell\le L}\max_{1\le k\le p}
\mathbb P_{n,\ell}(\tilde Z_k^2\Delta_{g,\ell}^2),
\qquad
B_{g,Z,n}
:=\max_{1\le\ell\le L}\max_{i\in I_\ell}\max_{1\le k\le p}
|\tilde Z_{ik}\Delta_{g,\ell}(X_i)|.
\]
Then
\[
\Delta_{g,Z,n}=o_p(1),
\qquad
B_{g,Z,n}\sqrt{\frac{\ell_p}{n}}=o_p(1).
\]

\item \emph{Dimension and sparsity growth.}
$\ell_d/n\to0$ and
\begin{equation}\label{eq:primitive-pointwise-growth}
\frac{s_S\sqrt{\ell_p\ell_d}}{\sqrt n}\to0,
\qquad
\frac{s_j\ell_p}{\sqrt n}\to0.
\end{equation}
\end{enumerate}
\end{assumption}

\begin{lemma}[Conditional-Bernstein bound for the baseline-induced Lasso score]
\label{lem:baseline-lasso-score}
Under Assumptions~\ref{as:rct}, \ref{as:reg}(P1), and
\ref{as:primitive-pointwise}(Q3),
\[
\left\|\mathbb P_n(\tilde ZR_S)\right\|_\infty
=O_p\!\left(
\Delta_{g,Z,n}\sqrt{\frac{\ell_p}{n}}
+B_{g,Z,n}\frac{\ell_p}{n}
\right)
=o_p\!\left(\sqrt{\frac{\ell_p}{n}}\right).
\]
Consequently, because $\ell_p\le\ell_d$, this baseline-induced score is
$o_p(\sqrt{\ell_d/n})$.
\end{lemma}

\begin{proof}
For a fold $\ell$ and coordinate $k$, write
\[
T_{\ell kS}:=\frac1n\sum_{i\in I_\ell}\tilde Z_{ik}R_{iS}
=-\frac1n\sum_{i\in I_\ell}
\tilde Z_{ik}\Delta_{g,\ell}(X_i)q_S(A_i).
\]
For this fold, set
\[
\Delta_{\ell,n}^2:=\max_{k\le p}
\mathbb P_{n,\ell}(\tilde Z_k^2\Delta_{g,\ell}^2),
\qquad
B_{\ell,n}:=\max_{i\in I_\ell}\max_{k\le p}
|\tilde Z_{ik}\Delta_{g,\ell}(X_i)|.
\]
Both quantities are $\mathcal G_\ell$-measurable, and they are bounded above by
$\Delta_{g,Z,n}^2$ and $B_{g,Z,n}$, respectively.
Conditional on $\mathcal G_\ell$, the summands are independent and mean zero by
Lemma~\ref{lem:global}. Their conditional variance sum and envelope satisfy
\[
\sum_{i\in I_\ell}
\E\{\tilde Z_{ik}^2R_{iS}^2\mid\mathcal G_\ell\}
\le \kappa_2 n_\ell\Delta_{\ell,n}^2,
\qquad
\max_{i\in I_\ell}|\tilde Z_{ik}R_{iS}|
\le C_qB_{\ell,n}.
\]
Since $n_\ell/n$ is bounded above and away from zero and $L$ is fixed, conditional
Bernstein's inequality gives, for every $t>0$ and a constant $C$ independent of $k$,
\[
\PP\!\left(
|T_{\ell kS}|>
C\left\{\Delta_{\ell,n}\sqrt{\frac{t}{n}}
+B_{\ell,n}\frac{t}{n}\right\}
\,\middle|\,\mathcal G_\ell
\right)
\le 2e^{-t}.
\]
Taking $t$ to be a sufficiently large multiple of $\ell_p=\log(2p)$, applying a union
bound over $k\in[p]$, replacing the fold-specific quantities by their maxima, and then
summing the fixed number of fold contributions yields the first display. The two conditions
in (Q3) make its first and second terms, respectively,
$o_p(\sqrt{\ell_p/n})$.
\end{proof}

\begin{proposition}[Primitive verification of Assumption~\ref{as:reg}(P4)]
\label{prop:primitive-pointwise}
Suppose Assumptions~\ref{as:rct}, \ref{as:reg}(P1), and
\ref{as:primitive-pointwise} hold. Set
\[
\lambda=A_\theta\sqrt{\frac{\ell_d}{n}},
\qquad
\lambda_{j,n}=A_\gamma\sqrt{\frac{\ell_p}{n}},
\]
where $A_\theta,A_\gamma$ are fixed and sufficiently large as functions only of
$K_Z,K_\varepsilon,C_q,c_\Sigma,C_\Sigma,L,$ and the limiting fold proportions. Then
\[
\|\hat\theta_S-\theta_S^\ast\|_1
=O_p\!\left(s_S\sqrt{\frac{\ell_d}{n}}\right),
\qquad
\|\hat\gamma_j-\gamma_j\|_1
=O_p\!\left(s_j\sqrt{\frac{\ell_p}{n}}\right),
\]
\[
\left\|\mathbb P_n(\tilde Z_{-j}\varepsilon_S)\right\|_\infty
=O_p\!\left(\sqrt{\frac{\ell_p}{n}}\right),
\qquad
\left\|\mathbb P_n(\tilde Z_{-j}\tilde Z_j)\right\|_\infty=O_p(1).
\]
Consequently, Assumption~\ref{as:reg}(P4) holds with
\[
r_{\theta,n}=s_S\sqrt{\frac{\ell_d}{n}},\quad
r_{\gamma,n}=s_j\sqrt{\frac{\ell_p}{n}},\quad
r_{\varepsilon,n}=\sqrt{\frac{\ell_p}{n}},\quad
b_{Z,n}=1,
\]
and the three products in \eqref{eq:pointwise-rates} vanish by
\eqref{eq:primitive-pointwise-growth}.
\end{proposition}

\begin{proof}
Under (Q1)--(Q2), products of the relevant sub-Gaussian variables are sub-exponential.
Indeed, the projection identity
$\E(Z_{-j}Z_{-j}^\top)\gamma_j=\E(Z_{-j}Z_j)$ and the eigenvalue bounds imply
$\|\gamma_j\|_2=O(1)$, so $V_j=Z_j-Z_{-j}^\top\gamma_j$ has a uniformly bounded
sub-Gaussian norm.
Bernstein's inequality and a union bound over $p$ coordinates give
\[
\left\|\mathbb P_n(\tilde Z\varepsilon_S)\right\|_\infty
=O_p\!\left(\sqrt{\frac{\ell_p}{n}}\right),
\qquad
\left\|\mathbb P_n(\tilde Z_{-j}V_j)\right\|_\infty
=O_p\!\left(\sqrt{\frac{\ell_p}{n}}\right).
\]
The empirical centering terms are of smaller order. Lemma~\ref{lem:baseline-lasso-score} and
$\ell_p\le\ell_d$ therefore imply
\[
\left\|\mathbb P_n\!\left[
\tilde Z\{\tilde\psi_S-\tilde Z^\top\theta_S^\ast\}
\right]\right\|_\infty
=O_p\!\left(\sqrt{\frac{\ell_d}{n}}\right).
\]
Because $M$ is fixed, $\ell_d\asymp\ell_p$, and (Q4) implies the usual
$s_S\ell_d/n\to0$ and $s_j\ell_p/n\to0$ sample-size requirements. The sub-Gaussian design
and population eigenvalue bounds in (Q1) therefore give the sample compatibility bounds on
the $s_S$- and $s_j$-sparse cones. The standard
Lasso oracle inequalities then yield
\[
\|\hat\theta_S-\theta_S^\ast\|_1
=O_p(s_S\lambda),
\qquad
\|\hat\gamma_j-\gamma_j\|_1
=O_p(s_j\lambda_{j,n}),
\]
which are the first two asserted rates. The same sub-exponential concentration gives the
asserted $r_{\varepsilon,n}$ bound and
$\|\mathbb P_n(\tilde Z_{-j}\tilde Z_j)\|_\infty=O_p(1)$. Substitution of these four rates
into \eqref{eq:pointwise-rates} gives
\[
\sqrt n\,\lambda_{j,n}r_{\theta,n}
\lesssim \frac{s_S\sqrt{\ell_p\ell_d}}{\sqrt n},
\qquad
\sqrt n\,r_{\gamma,n}r_{\varepsilon,n}
\lesssim \frac{s_j\ell_p}{\sqrt n},
\]
and $r_{\gamma,n}b_{Z,n}=o(1)$ follows from the second condition in
\eqref{eq:primitive-pointwise-growth}. This proves (P4).
\end{proof}

\begin{corollary}[A supremum-norm route for the cross-fitted baseline]
\label{cor:baseline-score-primitive}
Suppose Assumptions~\ref{as:rct}, \ref{as:reg}(P1), and
\ref{as:primitive-pointwise}(Q1),(Q2),(Q4) hold. Let
\[
a_{g,n}:=\max_{1\le\ell\le L}\|\Delta_{g,\ell}\|_\infty,
\qquad
\|h\|_\infty:=\sup_{x\in\mathcal X}|h(x)|,
\]
and assume
\[
a_{g,n}=o_p(1),
\qquad
\frac{\ell_p\ell_{np}}{n}\longrightarrow0.
\]
Then Assumption~\ref{as:primitive-pointwise}(Q3) holds. Consequently,
Proposition~\ref{prop:primitive-pointwise} applies, and the pointwise weighted-baseline
condition \eqref{eq:weighted-baseline} also holds.
\end{corollary}

\begin{proof}
Condition (Q1), $\ell_p/n\to0$, and the fixed number of folds imply
\[
\max_{\ell\le L}\max_{k\le p}
\mathbb P_{n,\ell}(\tilde Z_k^2)=O_p(1),
\qquad
\max_{i\le n}\max_{k\le p}|\tilde Z_{ik}|=O_p(\sqrt{\ell_{np}}).
\]
Consequently,
\[
\Delta_{g,Z,n}\le
a_{g,n}\max_{\ell,k}
\{\mathbb P_{n,\ell}(\tilde Z_k^2)\}^{1/2}=o_p(1)
\]
and
\[
B_{g,Z,n}\sqrt{\frac{\ell_p}{n}}
\le a_{g,n}\max_{i,k}|\tilde Z_{ik}|
\sqrt{\frac{\ell_p}{n}}
=o_p(1).
\]
Thus (Q3) holds. Moreover,
$\max_\ell\mathbb P_{n,\ell}\Delta_{g,\ell}^2\le a_{g,n}^2=o_p(1)$. The nodewise objective at
$\hat\gamma_j$ is no larger than at zero, so
$\mathbb P_n\hat V_j^2\le\mathbb P_n\tilde Z_j^2=O_p(1)$; fixed positive fold proportions
then give $\max_\ell\mathbb P_{n,\ell}\hat V_j^2=O_p(1)$. Hence
$\max_\ell\mathbb P_{n,\ell}(\hat V_j^2\Delta_{g,\ell}^2)=o_p(1)$, proving
\eqref{eq:weighted-baseline}.
\end{proof}

\paragraph{Why estimation requires more.}
Assumption~\ref{as:primitive-pointwise}(Q3) is stronger than the weighted consistency needed
for the final pointwise influence-function remainder because it controls the maximum of the
$p$ score-Lasso coordinates at the $\sqrt{\ell_p/n}$ scale. The observation-coordinate
logarithm $\ell_{np}$ enters only through the simple supremum-norm route used to bound the
sample envelope; it does not enlarge the Lasso penalty. Corollary~\ref{cor:baseline-score-primitive}
is only one sufficient route. More generally, (Q3) allows an unbounded baseline error whenever
its empirical weighted variance and envelope satisfy the displayed conditional-Bernstein
budgets. Thus no prescribed polynomial rate is needed for the pointwise influence-function
remainder, while the baseline contribution to the high-dimensional Lasso score remains below
the chosen penalty.


\section{Sufficient conditions for simultaneous inference}
\label{sec:supp-uniform-sufficient}

The simultaneous theorems are stated in terms of the uniform scale and correlation requirements
actually used by the Gaussian and bootstrap approximations. The following two lemmas give
convenient routes to the estimated-score correlation condition and to the baseline budget,
respectively; the subsequent corollary combines them with bounded sparse-regression
conditions. Throughout this section, write $\ell_p:=\log(2p)$ and retain
$\ell_n=\log(2nd_n)$ from the main text.

\begin{lemma}[A sufficient empirical-$L_2$ condition for
Assumptions~\ref{as:studentization} and \ref{as:score-cov}]
\label{lem:score-cov-sufficient}
Suppose Assumption~\ref{as:unif-rem}(U1) holds.
For an array $a=(a_i)_{i=1}^n$, write
$\|a\|_n=(\mathbb P_na_i^2)^{1/2}$. Define
\[
\Delta_{\mathrm{score},n}:=\max_{h\in\cH_n}\|\hat s_h-u_h\|_n,
\quad
m_n:=\max_{h\in\cH_n}|\mathbb P_nu_h|,
\quad
\Delta_n^0:=\max_{h,h'\in\cH_n}
|\mathbb P_n(\xi_h\xi_{h'})-\Sigma_{\xi,h,h'}|.
\]
If
\begin{equation}\label{eq:l2-score-sufficient}
\ell_n^2(\Delta_{\mathrm{score},n}+m_n+\Delta_n^0)=o_p(1),
\end{equation}
then Assumptions~\ref{as:studentization} and \ref{as:score-cov} hold.
\end{lemma}

\begin{proof}
Let $d_{i,h}:=\hat s_{i,h}-u_{i,h}$. Since
$\hat u_{i,h}=\hat s_{i,h}-\mathbb P_n\hat s_h$,
\[
\hat u_{i,h}-u_{i,h}
=d_{i,h}-\mathbb P_nd_h-\mathbb P_nu_h.
\]
The empirical Cauchy--Schwarz inequality therefore gives
\begin{equation}\label{eq:supp-centered-score-distance}
\max_{h\in\cH_n}\|\hat u_h-u_h\|_n
\le 2\Delta_{\mathrm{score},n}+m_n=:a_{\mathrm{score},n}.
\end{equation}
By the diagonal part of $\Delta_n^0$ and the uniform bounds on $\nu_h$ in (U1),
\[
\max_h|\mathbb P_nu_h^2-\nu_h^2|
\le C\Delta_n^0,
\qquad
\max_h\|u_h\|_n=O_p(1).
\]
Consequently,
\begin{align*}
\max_h|\hat\nu_h^2-\nu_h^2|
&\le
\max_h\|\hat u_h-u_h\|_n
\max_h\{\|\hat u_h-u_h\|_n+2\|u_h\|_n\}
+C\Delta_n^0\\
&=O_p(a_{\mathrm{score},n}+\Delta_n^0)=o_p(\ell_n^{-2}).
\end{align*}
The lower bound on $\nu_h$ now implies
$\PP(\min_h\hat\nu_h>0)\to1$ and
\[
\ell_n\max_h\left|\frac{\hat\nu_h}{\nu_h}-1\right|=o_p(1),
\]
which is Assumption~\ref{as:studentization}.

On the event of positive estimated scales, set
$\hat\xi_{i,h}:=\hat u_{i,h}/\hat\nu_h$. The preceding bounds and
\eqref{eq:supp-centered-score-distance} yield
\[
\max_h\|\hat\xi_h-\xi_h\|_n
=O_p(a_{\mathrm{score},n}+\Delta_n^0)=o_p(\ell_n^{-2}),
\qquad
\max_h(\|\hat\xi_h\|_n+\|\xi_h\|_n)=O_p(1).
\]
For every $h,h'$, another application of empirical Cauchy--Schwarz gives
\begin{align*}
|\mathbb P_n(\hat\xi_h\hat\xi_{h'})-
  \mathbb P_n(\xi_h\xi_{h'})|
&\le
\|\hat\xi_h-\xi_h\|_n\|\hat\xi_{h'}\|_n
+\|\xi_h\|_n\|\hat\xi_{h'}-\xi_{h'}\|_n.
\end{align*}
Taking the maximum and adding $\Delta_n^0$ proves
$\ell_n^2\Delta_{\Sigma,n}=o_p(1)$, which is
Assumption~\ref{as:score-cov}.
\end{proof}

\begin{lemma}[A population-$L_2(P_X)$ route to the uniform baseline budget]
\label{lem:uniform-baseline-population}
Suppose the fold and cross-fitting conditions in Assumption~\ref{as:unif-rem}(U1) hold. Define
\[
\Delta_{g,0,n}^2
:=\max_{1\le\ell\le L}\mathbb P_{n,\ell}\Delta_{g,\ell}^2,
\qquad
H_{g,n}:=\max_{1\le\ell\le L}\max_{i\in I_\ell}|\Delta_{g,\ell}(X_i)|.
\]
If
\[
\ell_n^2\max_{1\le\ell\le L}\|\Delta_{g,\ell}\|_{L_2(P_X)}=o_p(1),
\qquad
\max_{1\le\ell\le L}\|\Delta_{g,\ell}\|_\infty=O_p(1),
\]
then $\ell_n^2\Delta_{g,0,n}=o_p(1)$ and $H_{g,n}=O_p(1)$; hence condition (R3) below
holds.
\end{lemma}

\begin{proof}
Write $b_{\ell,n}:=\|\Delta_{g,\ell}\|_{L_2(P_X)}$. Conditional on the training sample for fold
$\ell$, the validation covariates are i.i.d.\ from $P_X$ and $\Delta_{g,\ell}$ is fixed. Hence, for
every $\epsilon>0$, conditional Markov's inequality gives
\[
\PP\!\left(
\ell_n^2\{\mathbb P_{n,\ell}\Delta_{g,\ell}^2\}^{1/2}>\epsilon
\,\middle|\,\mathcal F_{-\ell},I_1,\ldots,I_L
\right)
\le
\min\!\left\{1,\frac{\ell_n^4b_{\ell,n}^2}{\epsilon^2}\right\}.
\]
The right-hand side converges to zero in probability by the first displayed assumption and is
bounded by one, so its expectation also converges to zero. A union bound over the fixed number
of folds proves $\ell_n^2\Delta_{g,0,n}=o_p(1)$. Moreover,
$H_{g,n}\le\max_\ell\|\Delta_{g,\ell}\|_\infty=O_p(1)$, proving the second assertion and (R3).
\end{proof}

\begin{corollary}[A bounded-envelope primitive route to the simultaneous theorem]
\label{cor:primitive-uniform}
Assume Assumption~\ref{as:rct}, fixed folds as in Assumption~\ref{as:reg}(P1), and the following
conditions. All constants below are independent of $n,p,j,$ and $S$, but may depend on the
fixed factorial design.
\begin{enumerate}[leftmargin=1.8em,label=(R\arabic*)]
\item $\E Z=0$, $\E Z_k^2=1$,
$c_\Sigma\le\lambda_{\min}(\Sigma_Z)\le\lambda_{\max}(\Sigma_Z)\le C_\Sigma$, and
\[
\max_{k\le p}|Z_k|\le C_Z,
\qquad
\max_{j\le p}|V_j|\le C_V,
\qquad
\max_{S\in\cS}|\varepsilon_S|\le C_\varepsilon
\quad\text{almost surely}.
\]
The population normal equations
$\E(\varepsilon_S)=0$, $\E(Z\varepsilon_S)=0$, and $\E(Z_{-j}V_j)=0$ hold for every
$j,S$. Moreover, $c_\nu\le\min_{j,S}\nu_{jS}\le\max_{j,S}\nu_{jS}\le C_\nu$.
The fixed factorial design and positivity imply
$\max_{S\in\cS,a\in\cA}|q_S(a)|<\infty$ automatically.

\item With
$s_\theta=\max_S\|\theta_S^\ast\|_0$ and
$s_\gamma=\max_j\|\gamma_j\|_0$, use
\[
\lambda=A_\theta\sqrt{\frac{\ell_n}{n}},
\qquad
\lambda_j=A_\gamma\sqrt{\frac{\ell_p}{n}}
\quad(j\in[p]),
\]
for sufficiently large fixed constants $A_\theta,A_\gamma$. More generally, the conclusion is
unchanged when all nodewise penalties are bounded above and below by fixed positive multiples
of $\sqrt{\ell_p/n}$ and the lower multiple is sufficiently large.

\item Define
\[
\Delta_{g,0,n}:=
\max_{1\le\ell\le L}\{\mathbb P_{n,\ell}\Delta_{g,\ell}^2\}^{1/2},
\qquad
H_{g,n}:=\max_{1\le\ell\le L}\max_{i\in I_\ell}|\Delta_{g,\ell}(X_i)|.
\]
Then $\ell_n^2\Delta_{g,0,n}=o_p(1)$ and $H_{g,n}=O_p(1)$.

\item The growth conditions
\begin{equation}\label{eq:primitive-uniform-growth-a}
\frac{\ell_n^7}{n}\to0,
\qquad
s_\theta\ell_n\sqrt{\frac{\ell_p}{n}}\to0,
\qquad
s_\gamma\ell_n\sqrt{\frac{\ell_p}{n}}\to0,
\end{equation}
and
\begin{equation}\label{eq:primitive-uniform-growth-b}
\ell_n^2\left\{
\sqrt{\frac{s_\theta\ell_n}{n}}
+\sqrt{\frac{s_\gamma\ell_p}{n}}
\right\}\to0
\end{equation}
hold.
\end{enumerate}
Then Assumptions~\ref{as:unif-rem}, \ref{as:hd-array}, \ref{as:studentization}, and
\ref{as:score-cov} hold, and all
conclusions of Theorems~\ref{thm:uniform-gaussian}--\ref{thm:bootstrap-band} and
Corollary~\ref{cor:strong-fwer} follow.
\end{corollary}

\begin{proof}
We verify the four high-level assumptions in turn. Under (R1), bounded-variable Bernstein
inequalities and union bounds, together with the population eigenvalue bounds, give the usual
uniform Lasso and nodewise-Lasso rates
\[
\max_S\|\hat\theta_S-\theta_S^\ast\|_1
=O_p\!\left(s_\theta\sqrt{\frac{\ell_n}{n}}\right),
\qquad
\max_j\|\hat\gamma_j-\gamma_j\|_1
=O_p\!\left(s_\gamma\sqrt{\frac{\ell_p}{n}}\right).
\]
The baseline contribution to the first score is controlled conditionally fold by fold exactly
as in Lemma~\ref{lem:baseline-lasso-score}; (R3) makes it negligible relative to
$\sqrt{\ell_n/n}$. The remaining empirical quantities in (U2) may therefore be taken as
\[
\begin{gathered}
r_{\theta,n}=s_\theta\sqrt{\ell_n/n},
\qquad
r_{\gamma,n}=s_\gamma\sqrt{\ell_p/n},\\
r_{\varepsilon,n}=r_{V,n}=r_{Z,n}=r_{\tau,n}=\sqrt{\ell_p/n},
\qquad r_{\bar\varepsilon,n}=n^{-1/2},
\qquad b_{Z,n}=1,\\
\lambda_{\gamma,n}\asymp\sqrt{\ell_p/n}.
\end{gathered}
\]
Substitution shows that (R4) implies both
\eqref{eq:uniform-rate-budget} and \eqref{eq:uniform-denominator-budget}.

The nodewise prediction bound and (R1)--(R4) imply
$\max_{i,j}|\hat V_{ij}|=O_p(1)$. Hence
\[
\Delta_{g,V,n}=O_p(\Delta_{g,0,n}),
\qquad
B_{g,V,n}=O_p(H_{g,n}).
\]
Condition (R3), together with $\ell_n^7/n\to0$, now gives both bounds in
\eqref{eq:uniform-baseline-budget}. This verifies Assumption~\ref{as:unif-rem}.

Since $V_j$, $\varepsilon_S$, and $\nu_{jS}^{-1}$ are uniformly bounded under (R1), every
standardized influence coordinate $\xi_{i,h}$ is uniformly bounded. Thus
Assumption~\ref{as:hd-array} holds with a fixed $B_n$, and its growth condition is precisely
the first condition in \eqref{eq:primitive-uniform-growth-a}.

It remains to verify studentization and correlation estimation. The standard Lasso prediction
bounds, the nodewise prediction bounds, and (R3) yield
\[
\Delta_{\mathrm{score},n}
=O_p\!\left(
\sqrt{\frac{s_\theta\ell_n}{n}}
+\sqrt{\frac{s_\gamma\ell_p}{n}}
+\Delta_{g,0,n}
\right).
\]
Uniform boundedness and Bernstein's inequality also give
\[
m_n=O_p\!\left(\sqrt{\frac{\ell_n}{n}}\right),
\qquad
\Delta_n^0=O_p\!\left(\sqrt{\frac{\ell_n}{n}}\right).
\]
Equations~\eqref{eq:primitive-uniform-growth-a}--
\eqref{eq:primitive-uniform-growth-b} and (R3) therefore imply
$\ell_n^2(\Delta_{\mathrm{score},n}+m_n+\Delta_n^0)=o_p(1)$. Lemma~\ref{lem:score-cov-sufficient} proves
Assumptions~\ref{as:studentization} and \ref{as:score-cov}, completing the verification.
The stated conclusions then follow directly from
Theorems~\ref{thm:uniform-gaussian}--\ref{thm:bootstrap-band} and
Corollary~\ref{cor:strong-fwer}.
\end{proof}

\begin{remark}[Cost of the bounded-envelope route]\label{rem:envelope-cost}
Condition (R1) is not innocuous. Since
$\varepsilon_S=\{Y-g(X)\}\,q_S(A)-\alpha_S^\ast-Z^\top\theta_S^\ast$ and
$|Z^\top\theta_S^\ast|\le C_Z\|\theta_S^\ast\|_1$ almost surely, the uniform almost-sure bound
$\max_{S\in\cS}|\varepsilon_S|\le C_\varepsilon$ (with $C_\varepsilon$ independent of $n,p$)
effectively requires (i) $|Y-g(X)|$ bounded almost surely, and (ii)
$\sup_n\max_{S\in\cS}\{|\alpha_S^\ast|+\|\theta_S^\ast\|_1\}<\infty$. Similarly,
$|V_j|=|Z_j-Z_{-j}^\top\gamma_j|\le C_Z(1+\|\gamma_j\|_1)$ shows that
$\max_j|V_j|\le C_V$ is implied by (iii) $\sup_n\max_j\|\gamma_j\|_1<\infty$, which absent
further structure is also the natural checkable sufficient condition. Sparsity alone does not
deliver (ii)--(iii): the eigenvalue bound gives only $\|\gamma_j\|_2\le c_\Sigma^{-1/2}$, hence
$\|\gamma_j\|_1\le\sqrt{s_j}\,c_\Sigma^{-1/2}$, which stays bounded only if $s_j=O(1)$ or the
coefficients decay summably. When (i)--(iii) fail, Corollary~\ref{cor:primitive-uniform} is not
available and Assumptions~\ref{as:unif-rem}, \ref{as:hd-array},
\ref{as:studentization}, and \ref{as:score-cov} must be verified directly. We do not develop a
sub-Gaussian (unbounded-envelope) primitive route in
this paper: under sub-Gaussian tails the CCK moment bound has $B_n$ growing logarithmically,
which tightens the dimensional-growth and studentization requirements without changing the
argument's structure.
\end{remark}

\end{document}